%% file: LeFloch-Ma-CQG-14-Decembre.tex
\documentclass[10pt]{iopart}
\usepackage{iopams}  
\usepackage{slashed}
\usepackage{graphicx}
\usepackage{booktabs,rotating,tabularx}
\usepackage{mathrsfs}  
 
 \expandafter\let\csname equation*\endcsname\relax

\expandafter\let\csname endequation*\endcsname\relax

\usepackage{amsthm}
\usepackage{amsmath, amsfonts, amssymb, amscd}
\usepackage{xcolor} 

\newtheorem{theorem}{Theorem}[section]
\newtheorem{proposition}[theorem]{Proposition}
\newtheorem{lemma}[theorem]{Lemma}

\newtheorem{definition}[theorem]{Definition}
\numberwithin{equation}{section}
\DeclareFontFamily{U}{BOONDOX-calo}{\skewchar\font=45 }
\DeclareFontShape{U}{BOONDOX-calo}{m}{n}{
<-> s*[1.05] BOONDOX-r-calo}{}
\DeclareFontShape{U}{BOONDOX-calo}{b}{n}{
<-> s*[1.05] BOONDOX-b-calo}{}
\DeclareMathAlphabet{\mathcalboondox}{U}{BOONDOX-calo}{m}{n}
\SetMathAlphabet{\mathcalboondox}{bold}{U}{BOONDOX-calo}{b}{n}
\DeclareMathAlphabet{\mathbcalboondox}{U}{BOONDOX-calo}{b}{n}

\usepackage{chngcntr}
\counterwithin{figure}{section}

\input our-macros-CQG.tex

\newcommand \be {\begin{equation}}

\newcommand\ee {\end{equation}}
\newcommand \gbf{\mathbf g}
\newcommand \bse {\begin{subequations}}
\newcommand \ese {\end{subequations}}
\let\oldaligned\aligned
\def\aligned{\oldaligned\relax}

\newcommand \bei {\begin{itemize}} 
\newcommand\eei {\end{itemize}}

\begin{document}

\title[Nonlinear stability of self-gravitating massive fields]{Nonlinear stability of self-gravitating massive fields. 
\\ 
A wave-Klein-Gordon model}

%

\author{Philippe G. LeFloch$^1$ and and Yue Ma$^2$}

\address{$^1$ Laboratoire Jacques-Louis Lions and Centre National de la Recherche Scientifique, Sorbonne Universit\'e, 
4 Place Jussieu, 75252 Paris, France. 
\\
$^2$ School of Mathematics and Statistics, Xi'an Jiaotong University, Xi'an, 710049 Shaanxi, People's Republic of China.
}
\ead{contact@philippelefloch.org, yuemath@mail.xjtu.edu.cn}
\vspace{10pt}
\begin{indented}
\item[] December 2022 
\end{indented}

\begin{abstract} Significant advances were made in recent years on the global evolution problem for self-gravitating
massive matter in the small-perturbative regime close to Minkowski spacetime. To study the coupling between a Klein-Gordon equation and Einstein's field equations, we introduced the ``Euclidean-hyperboloidal foliation method'', which is based on the construction of a spacetime foliation adapted to the derivation of sharp decay estimates for wave and Klein-Gordon equations 
in a curved spacetime. We give here an outline of this method, together with a full proof for a wave-Klein-Gordon model which retains  some main  challenges arising with the Einstein-matter system. 
\end{abstract}


{\small

\setcounter{tocdepth}{1}

\tableofcontents

}

%
%
%

\section{Introduction}

\subsection{Global evolution of self-gravitating massive field}

\paragraph{Einstein-matter system.}

We are interested in four-dimen\-sional spacetimes $(\M, \gbf)$ where $\M$ is the manifold $\M \simeq [0, + \infty) \times \RR^3$, and $\gbf$ is a Lorentzian metric with signature $(-, +, +, +)$. The  Levi-Civita connection of this metric is denoted by $\nabla$ from which we determine the Ricci curvature tensor $\textbf{Ric}$ and the scalar curvature $\mathbf{R}$, respectively. The components of tensors such as $\gbf$ and $\textbf{Ric}$ are denoted\footnote{Throughout, Greek indices describe $0, 1,2,3$ and we use the standard convention of implicit summation over repeated indices, as well as raising and lowering indices with respect to the metric $g_{\alpha\beta}$ and its inverse denoted by $g^{\alpha\beta}$.} 
 by $g_{\alpha\beta}$ and $R_{\alpha\beta}$, respectively.  
In this notation, we impose 
Einstein's field equations 
\begin{equation} \label{eq-1-einstein-massif}
\Gbf= 8\pi \, \Tbf \quad \mbox{ in } (\Mscr,\gbf), 
\end{equation}
in which the left-hand side is Einstein's curvature tensor $\Gbf := \Rbf - {1 \over 2} \Rbf \, \gbf$. In the right-hand side, the energy-momentum tensor $\Tbf$ with components $T_{\alpha\beta}$ depends upon the nature of the matter  
under consideration and, specifically, we are interested in real-valued, {massive scalar fields} $\phi: \Mscr \to \RR$ described by the energy-momentum tensor
\begin{equation} \label{eq:Talphabeta}
T_{\alpha\beta} := \nabla_\alpha \phi \nabla_\beta \phi - \Big( {1 \over 2} \nabla_\gamma \phi \nabla^\gamma \phi + U(\phi) \Big) g_{\alpha\beta}. 
\end{equation}
Here, the potential $U=U(\phi)$ is a prescribed real-valued function satisfying 
\begin{equation} \label{eq:Uofphi}
U(\phi) = {1\over 2} c^2  \phi^2 + \Ocal(\phi^3)
\end{equation}
for some constant $c>0$, referred to as the {\sl mass} of the scalar field.

From the (twice contracted) Bianchi identities applied to \eqref{eq-1-einstein-massif} we deduce the matter evolution equations  
\begin{equation} \label{Eq1-15bis}
\nabla^\alpha T_{\alpha\beta} =0 \quad \mbox{ in } (\Mscr, \gbf).
\end{equation}
By denoting the wave operator by $\Box_g := \nabla^\alpha \nabla_\alpha$, it follows that the field $\phi$ satisfies a {\sl nonlinear Klein-Gordon equation} in a curved spacetime:
\begin{equation} \label{eq-KGG}
\Box_g \phi - U'(\phi) = 0 \quad \mbox{ in } (\Mscr, \gbf). 
\end{equation}
 For instance, with the choice $U(\phi)={1\over 2} c^2 \phi^2$,  \eqref{eq-KGG} is nothing but the linear Klein-Gordon equation $\Box_g \phi - c^2 \phi = 0$.
 
The initial value problem of interest here is formulated geometrically by prescribing an initial data set that is close to data associated with a asymptotically Euclidean, spacelike hypersurface of the (vacuum)  Minkowski spacetime.
 For suitably regular initial data, it is known that the Einstein equations \eqref{eq-1-einstein-massif} together with \eqref{eq-KGG} uniquely determines ``locally in time'' the spacetime geometry and the evolution of the matter field. 
Our challenge is precisely to the {\sl global evolution problem} and to investigate the global evolution of a massive matter field in the near-Minkowski regime. 


\paragraph{Wave-Klein-Gordon formulation.}

The equations under consideration are geometric in nature, and the degrees of gauge freedom must be fixed before tackling the nonlinear stability problem of interest by techniques of mathematical analysis. We assume the existence of global coordinate functions $x^\alpha: \Mscr \to \RR$ satisfying the wave gauge conditions ($\alpha=0,1,2,3$)
\begin{equation} \label{eq:wcooE}
\Box_g x^\alpha = 0. 
\end{equation}
In this gauge, the Einstein equations \eqref{eq-1-einstein-massif} take the form a nonlinear wave system of second-order partial differential equations, supplemented with second-order differential constraints. The main unknowns are then the metric coefficients  $g_{\alpha\beta}$ in the chosen coordinates, together with the real-valued field $\phi$. It is well-known that the constraints are preserved during the time evolution (cf., for instance, \cite{YCB}) and therefore it is sufficient to check them on the initial data set. 

Specifically, by introducing the modified wave operator $\BoxChapeau_{\gd} := \gd^{\alpha'\beta'} \del_{\alpha'} \del_{\beta'}$ (which takes the wave gauge into account), 
\eqref{eq-1-einstein-massif} and \eqref{eq:Talphabeta} can be restated as a nonlinear wave-Klein-Gordon system with unknowns $g_{\alpha\beta}$ and $\phi$, namely 
\begin{equation} \label{MainPDE-limit}
\aligned
\BoxChapeau_g g_{\alpha\beta} 
& = \Fbb_{\alpha\beta}(g, g;\del g,\del g) 
-16\pi \, \big( \del_{\alpha}\phi\del_{\beta}\phi + U(\phi)g_{\alpha\beta} \big),
\\
\BoxChapeau_g \phi  - U'(\phi) & = 0, 
\endaligned
\end{equation}
 supplemented  with the wave gauge conditions 
\begin{equation} \label{eq:gamnul3}
\aligned
\Gamma^\alpha & :=   g^{\alpha\beta} \Gamma_{\alpha\beta}^\lambda = 0, 
\qquad
\Gamma_{\alpha \beta}^{\lambda}
:= {1 \over 2} \,  g^{\lambda \lambda'} \big(\del_\alpha g_{\beta \lambda'}
+ \del_\beta g_{\alpha \lambda'} - \del_{\lambda'} g_{\alpha \beta} \big),
\endaligned
\end{equation}
together with Einstein's Hamiltonian and momentum constraints. We refer to~\cite{YCB} for this standard formulation.


\paragraph{Nonlinear stability theory.}  

Our main result for the Einstein-massive field~\cite{PLF-YM-main} is based on earlier partial advances in \cite{PLF-YM-CRAS}--\cite{PLF-YM-two} and establishes that initial data sets that are sufficiently close to (vacuum) Minkowski data
generates a global-in-time solution to the Einstein-Klein-Gordon system in wave gauge \eqref{MainPDE-limit}--\eqref{eq:gamnul3}. This global existence result for the set of partial differential equations (PDEs) translates into a geometric result,
 and the associated (globally hyperbolic) Cauchy development is proven to be future causally geodesically complete\footnote{That is, every affinely parameterized geodesic (of null or timelike type) can be extended toward the future (for all values of its affine parameter).} and, in fact, to approach the Minkowski geometry in all timelike, null, and spacelike directions. In other words, we prove that for a large family of initial data sets satisfying smallness conditions in {\sl energy and pointwise norms}, the matter field disperses in the infinite future and the formation of, for instance, black holes or gravitational singularities in the future development is avoided. 
 
 We emphasize that, as our project came under completion, we learned that Ionescu and Pausader~\cite{IP}--\cite{IP3}
 simultaneously solved the same problem by a different method, which is based on the notion of spacetime resonances.  

Let us recall that the nonlinear stability problem {\sl in the vacuum} was solved by 
Christodoulou and Klainerman via a gauge-invariant method \cite{CK}; see also Bieri~\cite{Bieri} for weaker decay conditions. 
Later on, Lindblad and Rodnianski discovered a proof in wave coordinates \cite{LR2}. For further contributions in the vacuum regime we refer to Hintz and Vasy~\cite{HintzVasy1,HintzVasy2}. 
On the other hand, the global dynamics of self-gravitating {\sl massive} matter fields has received far less attention, even in the regime of small perturbations of Minkowski spacetime. 
For other important contributions on various matter fields we refer to works by Bigorgne, Fajman, Joudioux, Lindblad, Smulevici, Taylor, and Wang~\cite{Bigorgne2,FJS,FJS3,LTay,Smulevici,Wang}. 
 

\subsection{Brief outlook on the method}
 
\paragraph{Bootstrap strategy.}

We rely on a bootstrap strategy which is based on the following arguments. 
\begin{itemize} 

\item[] {\sl Blow-up criterion.} A sufficiently regular, local-in-time solution cannot approach its {\sl maximal} time of existence, say $s^*$, at a time at which the energy (at a sufficiently high order) remains bounded. 
Otherwise, we would be able to extend this solution beyond $s^*$ by applying a local-in-time existence argument
 and this would contradict the fact that $s^*$ is chosen to be maximal.

\item[] {\sl Continuity criterion.}
The functional norms under consideration, which determine the regularity of the initial data, depend continuously upon the time variable, as long as a local-in-time solution exists. 

\item[] {\sl Formulating improved bound criterion.}
Suppose that on a time interval $[s_0,s_1]$ the solution satisfies inequalities expressed in terms of 
an energy functional at a (sufficiently) high-order of differentiation and, possibly, some other functionals of the solution. (This later part is irrelevant for the model problem treated next, but important in the treatment of the Einstein-Klein-Gordon system.)  Suppose that we can prove that these inequalities remain valid in a {\sl stronger} form with {\sl strictly smaller} constants. In these circumstances,  we deduce that the solution {\sl extends beyond} $s_1$.  

\item[] {\sl Deriving the improved bounds.}
Indeed, if $[s_0,s_1]$ is the maximal interval on which the set of inequalities holds, then at the ``final'' time $s_1$ thanks to the continuity criterion, at least one of the inequalities under consideration must become an {\sl equality.} However, in the case when we can prove that stronger inequalities holds on the same time interval, it follows that none of the inequalities can be an equality at $s_1$. This leads us to the conclusion that all of our inequalities hold for $s<s^*$, which is impossible when $s^*< + \infty$ in view of the 
blow-up criterion above.  

\end{itemize}


\paragraph{Euclidean-hyperboloidal framework.}

The basic features of our method are as follows. 

\begin{itemize}

\item {\it  Foliation.}  
Our framework is based on a foliation labelled by a parameter $s$, which is asymptotically Euclidean in the vicinity of spacelike infinity while timelike infinity is covered with slices that are asymptotically hyperboloidal. The foliation (cf.~\eqref{equ-foliation-def}, below)
is described via the introduction of a  coefficient $\xi=\xi(s,r)$ (cf.~\eqref{equation-xi-def}, below)
that interpolates between the interior domain
 in which $\xi= 1$ (for $r <\rhoH(s)$, a time-dependent radius) and an exterior domain
  in which $\xi = 1$ (for $r >\rhoE(s)$, a larger radius).  Here, we work in 
in a global coordinate chart $(t,x^a)$ with $a=1,2,3$, and $r^2 = \sum_a (x^a)^2$. 
  The two foliations are merged across a transition  region 
  associated with the interval $[\rhoH(s), \rhoE(s)]$. The time variable
  $s$ is connected to the standard Cartesian time $t$ 
  in such a way that is
   coincides with the standard hyperbolic time $\sqrt{t^2 - r^2}$ in the interior domain while it is of the order of $\sqrt{t}$ in the exterior.

\item {\it  Calculus rules and hierarchy.} To the proposed foliation we associated several frames, 
 including the semi-hyperboloidal frame $\delH$ (in \eqref{equa-shf}, below)  
 and the semi-null frame $\delN$ (in~\eqref{equa-snf}, below).   
 These vector fields are used to define
high-order operators and differentiate with the evolution equations of interest, as well as to decompose tensor fields such as the metric.  
The necessary calculus rules enjoyed by the vector fields and the operators of interest are provided in our theory  \cite{PLF-YM-main}
 as a series of technical lemmas.
  For instance, ordering properties allow us to work with ordered admissible operators $Z=\del^I L^J \Omega^K$, while commutator estimates are used to commute vectors fields with differential operators.  
A key {\sl hierarchy structure} enjoyed by quasi-linear commutators was uncovered, as stated in Propositions~\ref{lm 2 dmpo-cmm-H} and~\ref{prop1-12-02-2020-interior}. 
At this juncture, the notation $|u|_{p,k}$ (given in \eqref{equa-notation-pk}, below)
is very convenient to keep track of, both, the total order of differentiation (which we call the {\sl order,}  
denoted by $p$) and the total number of boosts or rotations (which we call the {\sl rank,} denoted by $k \leq p$). 
%


\item {\it  Functional inequalities.} New weighted Sobolev, Poincar\'e, and Hardy inequalities 
are required which are adapted to our Euclidean--hyperboloidal foliation.  This includes a Sobolev inequality for
 the hyperboloidal domain in Proposition~\ref{prop:glol-Soin} (which involves the boost vectors) 
 and a Sobolev inequality for the Euclidean-merging domain in Proposition~\ref{pro204-11-2}
 (which involves the distance to the light cone). 
The weighted Hardy inequality in Proposition~\ref{lem1-hardy}, which will be necessary in order to control undifferentiated terms such as metric coefficients. 
Further inequalities are required such as the Poincar\'e-type Proposition~\ref{propo-Poincare-ext}.
 

\item {\it  Pointwise decay of wave fields.}  The pointwise behavior of wave fields and their derivatives is based on an analysis of Kirchhoff's formula and establishes sharp estimates for solutions under assumptions on the source term. 
In Proposition~\ref{Linfini wave}, we distinguish between sub-critical, critical, and super-critical regimes, and 
we prove estimates with various decay behaviors in terms of the radial distance and the distance to the light cone. 
We also derive (cf.~Case 0 therein) a property within the light cone. 
The control of the Hessian of solutions to the wave equation must also be investigated at arbitrary order,  
and in our analysis we find it useful to 
distinguish between the near/far light cone regions and to rely
 on two different decompositions of the wave operator.  

 
\item {\it  Pointwise decay of Klein-Gordon fields on curved spacetimes.} We also derive sharp decay of solutions to Klein-Gordon fields and their derivatives. 
Specifically, we establish Propositions~\ref{prop1-23-11-2022-M} and \ref{lem 2 d-e-I}, below. 
 For a summary see also Proposition~\ref{lem 1 d-KG-e}, below.  

\end{itemize}   


\subsection{Proposed wave-Klein-Gordon model}  
 
\paragraph{Model of interest.}
 
We will illustrate some key features of our proof for the Einstein equation by focusing here on a model. The system under consideration now is formally derived from the 
Einstein-massive field system by suppressing null nonlinearities and quasi-null nonlinearities, and 
by replacing the metric by a scalar unknown.  
In turn, we obtain a wave equation and a Klein-Gordon equation coupled through quadratic terms
of zero, first, and second order, as follows: 
\begin{equation}\label{eq 1 model}
\aligned
- \Box u 
& =   P^{\alpha\beta} \del_\alpha \phi \del_\beta \phi + R \, \phi^2, 
\\
- \Box \phi + c^2 \, \phi
& =   H^{\alpha\beta} u \, \del_\alpha\del_\beta \phi. 
\endaligned 
\end{equation}
Here, the unknown functions $u=u(t,x)$ and $\phi=\phi(t,x)$ are defined in the 
future of a spacelike hypersurface o which initial data are prescribed. 
For simplicity in our presentation, the coefficients of the nonlinearities $P^{\alpha\beta}$, $R$, and $H^{\alpha\beta}$ are assumed to be constants, but a generalization to  non-constant coefficients would only require 
straightforward conditions on the derivatives.  The global existence theory for \eqref{eq 1 model} is presented in 
Section~\ref{section-strategy}, below; cf.~Theorem~\ref{theo-stable-model}. 
 

\paragraph{Outline of this paper.}

We begin, in Section~\ref{sect-geomEHF}, with the definition of the proposed Euclidean-Hyperboloidal foliation and
we state various properties concerning the vector frames of interest, commutator properties, and functional inequalities adapted to this foliation.  In Section~\ref{section-wkg}, we present several pointwise estimates for wave and Klein-Gordon equations posed on the Euclidean-Hyperboloidal foliation. In Section~\ref{section-strategy}, 
we give our main stability statement for the Einstein equations 
and we outline its proof by including also a key analysis of the null and quasi-null structure of the Einstein equations. 
 In Section~\ref{section-555}, we state our global existence result for the model. 
Section~\ref{section---666} is devoted to the proof of global existence for the model; we first derive energy and pointwise estimates in the Euclidean-merging domain and in the hyperboloidal domain, and we finally close the bootstrap by taking advantage of a hierarchy enjoyed by our estimates. 


Finally, let us summarize our main notation in the following table:  
$$
\aligned
& \aligned 
& \Mscr_s = \MH_s \cup \MM_s \cup \Mext_s
&& \mbox{ spacelike slices}
\\
& \Time(s,r)
&& \mbox{ global time function}
\\
& \xi(s,r) 
&& \mbox{ foliation coefficient}
\\   
& \zeta(s,r) 
&& \mbox{ energy coefficient}
\\  
& \delH_0 = \del_t, \quad \delH_a = \frac{x^a}{t} \del_t + \del_a
&& \mbox{ semi-hyperboloidal frame}
\\
& \delN_0 = \del_t, \quad \delN_a = \delsN = {x^a \over r} \del_t + \del_a
&& \mbox{ semi-null frame}
\\
& \delEH_s = ( \del_s T) \del_t,
\quad
\delsEH_a = \del_a + (x^a /r) (\del_r \Time) \del_t  
&& \mbox{ Euclidean--hyperboloidal frame} 
\\
& \crochet(s,r) 
&& \mbox{ energy weight } 
\\ 
\endaligned
\\
\endaligned 
$$


\section{Geometric properties of Euclidean-Hyperboloidal foliations}
\label{sect-geomEHF} 

\subsection{Geometry of the Euclidean-hyperboloidal slices}

\paragraph{Euclidean--hyperboloidal time function.} 

We begin by describing the foliation of interest. It is defined over a manifold $\Mscr \simeq 
\RR^{1+3}_+ = \RR_+ \times \RR^3$ covered by a global coordinate chart denoted by $(t,x): \Mscr \to  \RR^{1+3}_+$. For convenience, we assume $t \geq 1$ and we will impose initial data on the hypersurface $t=1$ and solve a Cauchy problem in the future of this initial hypersurface. 
We also set $x= (x^a) \in \RR^3$ and $r^2 = |x|^2 = \sum_{a=1}^3|x^a|^2$.  The ground state of the gravity theory is given by the Minkowski metric $g_{\Mink} = -dt^2 + \sum_a(dx^a)^2$, and we are interested in small perturbations of this metric. 

Our first task is defining a suitable spacetime foliation. 
To label our foliation, we use a (Euclidean-hyperboloidal) {\sl time parameter} $s \geq s_0$ and we distinguish between several domains on each hypersurface of constant $s$. To this end, we introduce the  
{\it hyperboloidal and Euclidean radii} at any time $s$ 
\be
\rhoH(s) := {1 \over 2} (s^2 -1), 
\qquad \rhoE(s) := {1 \over 2} (s^2 +1).
\ee
Let us consider any cut-off function $\chi: \RR \to [0,1]$ 
satisfying 
$
\chi(x) = 
\begin{cases}
0,   & x \leq 0, 
\\
1,   & x> 1, 
\end{cases}
$
and (for simplicity in some of our arguments) $\chi^{(m)}(x)  > 0$ for all $x \in (0,1/2)$ and each $m= 0,1,2,3$. By definition,  the 
{\it foliation coefficient} is the function 
\begin{equation}\label{equation-xi-def}
\xi(s,r) := 1-\chi(r- \rhoH(s)) 
= \begin{cases}
1, \quad & r <\rhoH(s), 
\\
0, \quad & r > \rhoE(s), 
\end{cases}
\ee
and will be applied to``select'' the hyperboloidal domain. 
 

We define the function $t = \Time(s,r)$ by solving the ordinary differential equation 
\be
\del_r \Time(s,r) = \frac{r \, \xi(s,r)}{(s^2+r^2)^{1/2}}, 
\qquad 
\Time(s,0) = s. 
\ee
It can be checked that this time function enjoys the following properties: 
\be
\Time(s,r) = \begin{cases} 
(s^2+r^2)^{1/2},   \quad &  r \leq \rhoH(s) \qquad \mbox{ (hyperboloidal domain),} 
\\
r + 1 = (s^2+1)/2, \quad &  r = \rhoH(s),
\\
\Time^\E(s), &  r \geq \rhoE(s) \qquad \mbox{ (Euclidean domain),}
\end{cases}
\ee
in which $\Time^\E=\Time^\E(s)$ is independent of $r$ and, for universal constants $K_1, K_2>0$,
\be
K_1 \, s^2 
\leq \Time(s,r) \leq K_2 \, s^2
 \qquad  
\text{ in }  \MME_{[s_0, + \infty)}, 
\ee
together with 
\be
\aligned
& 0\leq \del_r \Time(s,r)<1, \quad 
&& \mbox{ (slices of constant $s$ are spacelike),}
\\ 
& 0 < \del_r\Time(s,r)<1, \quad 
&& \mbox{ when $0<r \leq \rhoH(s)$},  
\\
& |\del_r \, \del_r \Time(s,r)| \lesssim 1. 
\endaligned
\ee

 \begin{figure}
\centerline{\includegraphics[width=10.5cm, height=2.5cm]{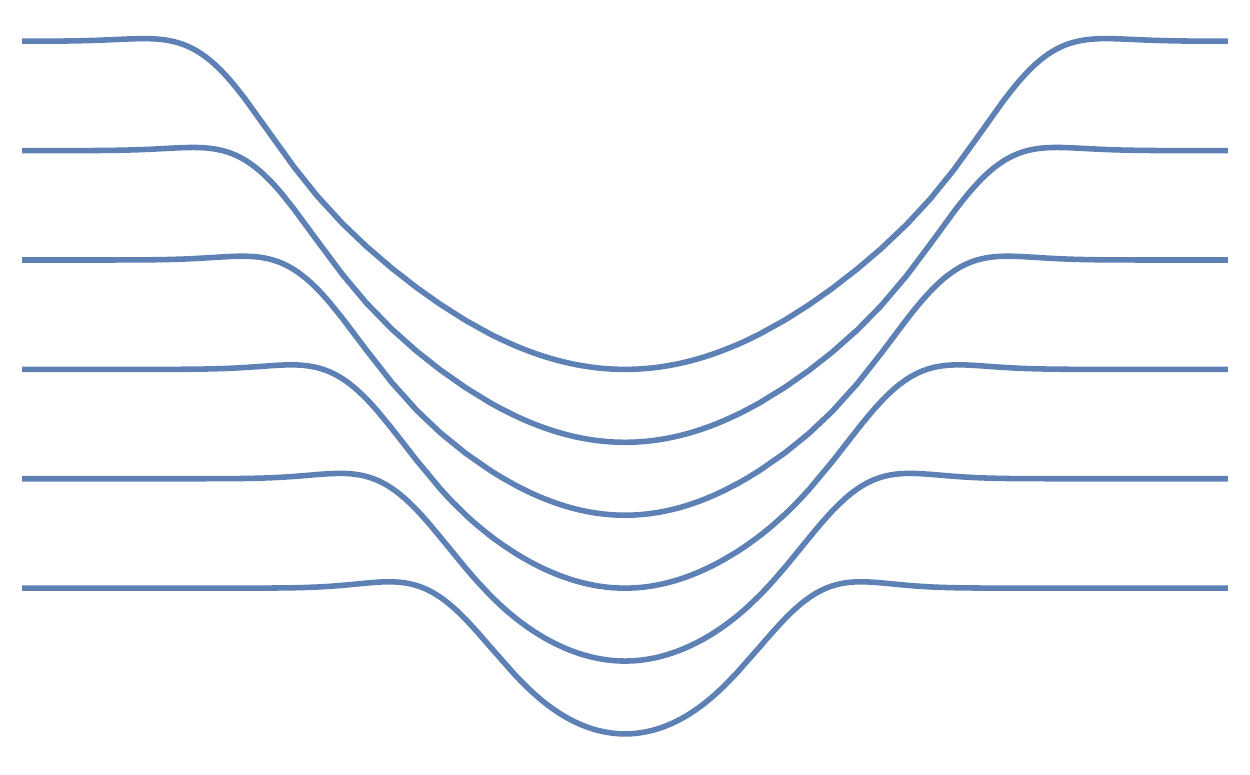}}
\caption{Spacetime foliation defined by merging 
\\
asymptotically Euclidean and asymptotically hyperboloidal slices}
\end{figure} 

\paragraph{Foliation of interest.}

A one-parameter family of spacelike, asymptotically Euclidean hypersurfaces is defined as 
\be
\Mscr_s := \big\{ (t,x^a)\in\Mscr \, \big/ \, t = \Time(s,r) \big\}. 
\ee
In the future of the initial surface $\{t=1\}$, namely 
$\{t\geq 1\} = \Mscr^{\init}\cup \bigcup_{s\geq s_0}\Mscr_s$, 
we distinguish the {\sl initial domain} (as we call it) $\Mscr^{\init} = \{(t,x)\,/\, 1\leq t\leq T(s_0,r)\}$ within which standard local-in-time existence arguments apply. 
Each slice $\Mscr_s = \MH_s \cup \Mtran_s \cup \Mext_s$  
is decomposed into three domains (with overlapping boundaries): 
\begin{equation}\label{equ-foliation-def} 
\aligned
\MH_s & :=   \big \{ t = \Time(s,|x|) \, \big/ \, 
|x| \leq \rhoH(s)
\big\}, 
&& \mbox{asymptotically hyperboloidal,}
\\
\MM_s & :=   \big\{t = \Time(s,|x| )  \, \big/ \, 
\rhoH(s)\leq  |x| \leq \rhoE(s)
\big\}, 
&&  \mbox{merging (or transition),}
\\
\Mext_s & :=   \big\{ t=\Time(s)  \, \big/ \, 
\rhoE(s) \leq |x|
\big\}, 
&& \mbox{asymptotically Euclidean}, 
\endaligned
\ee
with also $\MME_s := \Mext_s \cup \MM_s$.
Some additional notation is needed: 
\be
\aligned
\Mscr_{[s_0,s_1]} & :=   { \big\{ \Time(s_0,r)\leq t\leq \Time(s_1,r) \big\} = \bigcup_{s_0\leq s\leq s_1} \Mscr_s}, 
\qquad 
\Mscr_{[s_0, + \infty)} & := \bigcup_{s\geq s_0} \Mscr_s, 
\endaligned
\ee
and, similarly, we set $\Mscr^\H_{[s_0, s_1])}$, $\Mscr^\H_{[s_0, + \infty)}$, etc. 
By construction,  there exists a function $c=c(s)\in (0,1)$ such that the radial variable $r$ in each of the three domains satisfies 
\be
r = |x| \in  \,
\begin{cases}
\hskip1.cm [0, t-1],
& \MH_{[s_0, + \infty)},
\\ 
[t-1,t -  c(s)], 
\qquad & \MM_{[s_0, + \infty)}, 
\\
[t - c(s), +\infty),
& \Mext_{[s_0, + \infty)}. 
\end{cases}
\ee
The future-oriented normal and the volume element associated with the hypersurfaces are given by 
\bse
\be
\aligned
n_s & =  \frac{\big(1, - (x^a/r)\del_r \Time\big)}{\sqrt{1+|\del_r \Time|^2}} 
= \nu^{-1} \Big((s^2 + r^2)^{1/2},-x^a\xi(s,r) \Big), 
\\
\nu & = \big( s^2 +r^2(1+\xi(s,r)^2)\big)^{1/2}, 
\endaligned
\ee
together with 
\be
\aligned 
d\sigma_s &= (1+|\del_r \Time|^2) \, dx = \nu \, (s^2 + r^2)^{-1/2} \, dx, 
\\
n_sd\sigma_s 
& = (1, - (x^a/r)\del_r \Time) \, dx = (1, - \del_a \Time) \, dx = \Big(1,\frac{-\xi(s,r)x^a}{(s^2+r^2)^{1/2}}\Big) \, dx. 
\endaligned
\ee
\ese

 
\paragraph{Change of variables.}
 
We will be relying on two  parameterizations of the Euclidean-hyperboloidal hypersurfaces, namely in terms of the variables 
 $(t,x)$, or in terms of the variables $(s,x)$ determined by the function $T$. This latter function 
by construction is strictly increasing in $s$ and, in fact, $(s,x) \mapsto (t,x) = (\Time(s,x),x)$ is a smooth and global diffeomorphism. The Jacobian matrix associated with the Euclidean--hyperboloidal foliation  reads 
$
\left(
\begin{array}{cc}
\del_s \Time &\del_sx
\\
\del_x \Time &\del_x x 
\end{array}
\right)
=
\left(
\begin{array}{cc}
\del_s \Time &0
\\
(x^a/r) \, \del_r \Time & I
\end{array}
\right)
$
and the
Jacobian is $J(s,x) = \del_s\Time(s,x)$, leading to the corresponding volume element $dtdx = J \, dsdx$. 
It can be established that   
$$
J 
\leq
\begin{cases}
{s \over \Time}= s \, (s^2 +r^2)^{-1/2},   
\\
\xi s \, (s^2 +r^2)^{-1/2} + (1- \xi) \, 2s,
\\
2s,  
\end{cases}
\, 
J 
\geq
\begin{cases}
{s \over \Time} = s \, (s^2 +r^2)^{-1/2}, 
 \quad      & \MH_s,
\\
\xi \, s \, (s^2 + r^2)^{-1/2} + (1- \xi) 3s/5, 
 \quad      &   \MM_s, 
\\
3s/5 
 \quad      &  \Mext_s. 
\end{cases}
$$
  

\subsection{Frames of interest and commutators}

\paragraph{Frames of interest.}

We will combine estimates involving different frames, as follows. 
\bei 

\item The {\it semi-hyperboloidal frame} 
\be\label{equa-shf} 
\delH_0 = \del_t, 
\qquad
\delH_a = \delsH_a  = \frac{x^a}{t} \del_t + \del_a, 
\ee
and was already introduced by the authors in \cite{PLF-YM-book}. It is defined globally in $\Mscr_s$, relevant within the hyperboloidal domain and is relevant in order to 
(1) exhibit the (quasi-)null form structure of the nonlinearities 
and, in turn,  
(2)  establish decay properties in timelike and null directions. Some of our arguments also involve radial integration based on  
$ 
\delsH_r = (x^a /r)\delsH_a$.

\item 
The {\it semi-null frame}
\be\label{equa-snf} 
\delN_0  = \del_t, 
\qquad  
\delN_a = \delsN_a = {x^a \over r} \del_t + \del_a 
\ee
is defined everywhere in $\Mscr_s$ except on the center line $r=0$
and is the appropriate frame within the Euclidean-merging domain 
in order to (1) exhibit the structure of the (null, quasi-null) nonlinearities of the field equations, and 
(2) establish decay properties in spatial and null directions.

\item The {\it Euclidean--hyperboloidal frame}
\be
\delEH_0 = \del_t, 
\quad
\delEH_a  = \delsEH_a 
=   \del_a + (x^a/r)\del_r \Time \, \del_t  
\ee
involves tangent vectors to the slices $\Mscr_s$ and provides an 
 interpolation between $\delEH_a = \delH_a$ in $\MH_s$, and $\delEH_a = \del_a$ in $\Mext_s$. 
Some of our arguments are also based on radial integration based on $\delsEH_r = (x^a/r)\delsEH_a$. 

\eei
 
\noindent Moreover, various changes of frame formulas are useful such as
$\delN_\alpha = \PhiN{}_\alpha^\beta \, \del_\beta$ and $\del_\alpha = \PsiN_{\alpha}^\beta \delN_\beta$ with 
\be
\big(\PhiN_{\alpha}^\beta \big)
= 
\left(
\begin{array}{cccc}
1 &0 &0 &0
\\
x^1/r& 1 &0 &0
\\
x^2/r&0  &1 &0
\\
x^3/r& 0 &0 &1
\end{array}
\right), 
\qquad 
\big(\PsiN_{\alpha}^\beta \big)
= 
\left(
\begin{array}{cccc}
1 &0 &0 &0
\\
-x^1/r& 1 &0 &0
\\
-x^2/r&0  &1 &0
\\
-x^3/r& 0 &0 &1
\end{array}
\right).
\ee 


\paragraph{Commutator properties.}

It is convenient to rely on the following definition. 
Let $L_a = x^a\del_t + t\del_a$ and
 $\Omega_{ab} = x^a\del_b - x^b\del_a$ be the Lorentzian and Euclidean rotations. 
 We use $L^J, \Omega^K$ and $\del^I$ for high-order derivatives consisting of combinations of 
 $L_a, \Omega_{ab}$ and $\del_{\alpha}$ where $I,J,K$ denote ordered multi-indices. 

\begin{definition} An operator $Z=\del^I L^J \Omega^K$ is called an (ordered) {\it  admissible operator} and  
 for such an operator
  one associates its {\it  order} and {\it  rank} by  
\begin{equation}
\ord(Z) = |I|+|J| + |K|,  
\quad 
\rank(Z) = |J| + |K|,
\qquad 
Z = \del^I L^J \Omega^K.
\end{equation}
An operator $\Gamma = \del^IL^J\Omega^KS^l$ is called an {\it  ordered conformal operator}. Its {\it  order} and {\it  rank} are defined similarly:
$$
\ord(\Gamma) = |I|+|J| + |K| + l, 
\quad 
\rank(\Gamma) = |J| + |K| + l, 
\qquad 
\Gamma = \del^IL^J\Omega^KS^l.
$$
\end{definition}

The proof of the following statement can be found in~\cite[Part~1]{PLF-YM-main}. To deal with differential operators 
$\del^I L^J \Omega^K$ we use the notation $\ord(Z) = |I|+|J|+|K|$ (the order) and $\rank(Z) = |J|+|K|$ (the rank). 
Given two integers $k \leq p$, it is convenient to introduce the notation 
\be\label{equa-notation-pk}
|u|_{p,k} :=  \max_{\ord(Z) \leq p, \rank(Z) \leq k} |Z u|, 
\qquad 
|u|_{p} :=  \max_{\ord(Z) \leq p} |Z u|. 
\ee

\begin{lemma}[Estimates for linear commutators]
\label{prop-newpropo} 
For any admissible field $Z$ satisfying $\ord(Z) \leq p$ and $\rank(Z) \leq k$ one has 
\begin{subequations}
\begin{equation} \label{eq2-31-01-2020}
|[Z, \del] u | \lesssim |\del u|_{p-1,k-1},
\end{equation}
\begin{equation} \label{eq1-31-01-2020}
| [Z, \del \del ] u | 
\lesssim |\del\del u|_{p-1,k-1} \lesssim |\del u|_{p,k-1}.
\end{equation}
\end{subequations}
\end{lemma} 

In the following, it is convenient to denote by $\LOmega \in\{ L_a, \Omega_{ab}\}$ the collection of boosts and spatial rotations. 

\begin{proposition}[Hierarchy structure for quasi-linear commutators. Euclidean-merging domain]
\label{lm 2 dmpo-cmm-H} 
Let $Z$ be an admissible operator with ${\ord(Z) = p}$ and $\rank(Z) = k$ and let $H,u$ be functions defined in the  Euclidean-merging domain $\MME_s$.  Then, one has 
\begin{equation} \label{eq8-14-03-2021}
|[Z,H]u|\lesssim 
\sum_{k_1+p_2=p\atop k_1+k_2=k } | \LOmega H|_{k_1-1} |u|_{p_2,k_2}
+ \sum_{p_1+p_2=p\atop k_1+k_2=k} |\del H|_{p_1-1,k_1} |u|_{p_2,k_2}, 
\end{equation} 
\begin{equation} \label{eq7-14-03-2021}
\aligned
 |[Z,H\del_{\alpha}\del_{\beta}]u|
& \lesssim  \,  |H| \, |\del\del u|_{p-1,k-1} 
+ 
\sum_{k_1+p_2=p\atop k_1+k_2=k }
|\LOmega H|_{k_1-1} |\del\del u|_{p_2,k_2}
\\
&+
\sum_{p_1+p_2=p\atop k_1+k_2=k}   |\del H|_{p_1-1,k_1} |\del\del u|_{p_2,k_2}, 
\endaligned
\end{equation} 
\end{proposition}

\begin{proposition}[Hierarchy property for quasi-linear commutators. Hyperboloidal domain] 
\label{prop1-12-02-2020-interior} 
For any function $u$ defined in $\MH_{[s_0,s_1]}$ and for any admissible operator $Z$ with $\ord(Z) = p$ and $\rank(Z) = k$ one has 
\bse
\begin{equation}\label{equa-hsgd57} 
\aligned
\, 
& |[Z,H^{\alpha\beta} \del_\alpha \del_{\beta}]u|
\lesssim  T^\textbf{hier}  + T^\easy +  T^{\textbf{super}}, 
\endaligned
\end{equation}
\begin{equation}
\aligned
T^\textbf{hier} & 
|\Hu^{00} | \, | \del\del u|_{p-1,k-1} +\sum_{k_1+p_2=p\atop k_1+k_2=k} |L \Hu^{00} |_{k_1-1,k_1-1} |\del\del u|_{p_2,k_2}, 
\\
T^\easy
& :=
\sum_{p_1+p_2=p\atop k_1+k_2=k} |\del \Hu^{00} |_{p_1-1,k_1} |\del\del u|_{p_2,k_2}
+ t^{-1} |H| \, | \del u|_{p}, 
\\
T^{\textbf{super}} &
:= \sum_{p_1+p_2=p} |\delsH H|_{p_1-1} |\del u|_{p_2+1}
+ t^{-1} \!\!\!\!\sum_{p_1+p_2=p} \!\!\!\! |\del H|_{p_1-1} |\del u|_{p_2+1}.
\endaligned
\end{equation}
\ese
\end{proposition}


\subsection{Energy functionals}

\paragraph{Weight coefficients.}

The fundamental energy functional (stated shortly below) involved another geometric weight, denoted by $\zeta = \zeta(t,x)$ and defined by 
\be  
\zeta(s,r)^2 
 = 1 - \frac{r^2 \xi^2(s,r)}{s^2+r^2} = {s^2 \over s^2 +r^2} + (1- \xi^2(s,r)){r^2 \over s^2 +r^2}.
\ee
This weight coincides with  $s/t = s /(s^2+r^2)^{1/2}$ in the hyperboloidal domain, while it reduces to $1$ in the Euclidean domain. In fact, it provides us with an interpolation between the energy density induced on hyperboloids and the one induced on Euclidean slices. Various estimates on $\zeta$ can be proven, for instance 
$\frac{|r-t \, |+1}{r} \lesssim \zeta^2 \leq\zeta\leq 1$ valid within the Euclidean-merging domain, as well as
$K_1 \, \zeta^2 s\leq J\leq  K_2 \, \zeta^2 s$ in the merging domain (for some universal constants $K_1,K_2$). 
Furthermore, we introduce a weight which measures the distance to the light cone, and is defined from the prescription of a smooth and non-decreasing function $\aleph$ satisfying 
\be
\aleph(y) = 
\begin{cases}
0, \quad 
& y \leq -2, 
\\
y+2, \quad &  y \geq {-1}, 
\end{cases}
\ee 
and specifically we set 
$\crochet := 1 + \aleph({r-t})$,
which we refer to as the {\sl energy weight.} 
Recall that $\aleph'$ is non-negative. This choice is easier to work with, but an equivalent energy having a more geometric form can be based on the unknown metric $g$. 


\paragraph{Energy identity.}

We multiply the wave-Klein-Gordon equation $\Box u - c^2 u = F$ (with $c \geq 0$) by $- 2 \, \crochet^{2 \eta}  \del_t u$
with $\Box  = \Box_{\gMink}  = - \del_t \, \del_t + \sum_{a=1,2,3} \del_a\del_a$. 
We treat simultaneously the wave and Klein-Gordon operators by assuming here that $c \geq 0$.  
We find the divergence identity 
$$
\aligned
 \del_t \big( V^0_{\eta,c}[u] \big )+ \del_a \big( V^a_{\eta,c}[u] \big)
& =
2\eta \crochet^{-1} \aleph'({ r-t})(-1, x^a/r) \cdot V_{\eta,c}[u] 
- 2 \, \crochet^{2 \eta} \del_t u \, F 
\\
V_{\eta,c}[u]  
& = - \crochet^{2 \eta} \big(-|\del_t u|^2 - \sum_a|\del_au|^2 - c^2u^2, 2 \del_t u\del_au\big). 
\endaligned
$$
We define our energy functional on each Euclidean--hyperboloidal slice $\Mscr_s$ as
\be  
\aligned
\Eenergy_{\eta,c}(s,u) 
& = \int_{\Mscr_s}V_{\eta,c}[u] \cdot n_sd\sigma_s 
\\
&
= \int_{\Mscr_s} 
\Big(|\del_tu|^2 + \sum_a|\del_au|^2 + \frac{2x^a\xi(s,r)}{(s^2+r^2)^{1/2}} \del_t u\del_a u + c^2u^2 \Big) \, 
\crochet^{2\eta} \, dx 
\endaligned
\ee
or, equivalently,
\be  
\aligned
\Eenergy_{\eta,c}(s,u) 
& =  \int_{\Mscr_s} \Big(\zeta^2|\del_t u|^2 + \sum_a |\delsEH_au|^2 + c^2u^2 \Big) \,
\crochet^{2\eta} \, dx, 
\\
&
= \int_{\Mscr_s}
\Big(\zeta^2 \sum_a|\del_au|^2 +  \frac{\xi^2(s,r)}{s^2+r^2} \sum_{a<b} |\Omega_{ab}u|^2 
+ |\delsEH_r u|^2 + c^2 \, u^2 \Big) \, \crochet^{2\eta} \, dx. 
\endaligned
\ee
It involves the energy coefficient $\zeta$,   
which is non-trivial in the merging and hyperboloidal domains, and depends upon our choice of coefficient $\xi = \xi(s,r)$. 


The energy identity in curved spacetime associated with the wave or Klein-Gordon equation  $g^{\alpha\beta} \del_\alpha \del_\beta u - c^2 u = F$ is expressed by decomposing the curved metric $g$ as 
$
g^{\alpha\beta} =  g_\Mink^{\alpha\beta} + H^{\alpha\beta}. 
$
\bse
After defining the energy-flux vector (with $a =1,2,3$) 
\be  
V_{g, \eta,c}[u]
=
-\crochet^{2 \eta} \Big(
g^{00} |\del_t u|^2 - g^{ab} \del_au \, \del_bu - c^2u^2, \ 2 g^{a\beta} \del_t u  \del_\beta u \Big), 
\ee
which depends upon $g$ as well as the weight $\crochet^\eta$, we easily find 
the energy identity
\be
\aligned 
\dive  V_{g, \eta,c}[u] 
= -  \Omega_{g, \eta,c}[u] + G_{g, \eta}[u] - 2 \crochet^{2 \eta} \, \del_tu \, F, 
\endaligned
\ee 
in which 
\be\label{eq8-27-11-2022-M}
\aligned
\Omega_{g, \eta,c}[u] 
& = 
-2\eta \, \crochet^{-1} \aleph'({ r-t}) (-1,x^a/r) \cdot V_{g,\eta,c}[u]
\\
& = 2 \, \eta \, \crochet^{2\eta-1} \aleph'({ r-t}) \, \Big(g^{\N ab}\delsN_au\delsN_bu - H^{\N00} |\del_tu|^2 + c^2 u^2\Big), 
\\ 
G_{g, \eta}[u] 
& =  -  \del_tH^{00} | \crochet^\eta \del_t u|^2 + \del_tH^{ab} \crochet^{2 \eta} \del_au\del_b u 
- 2 \crochet^{2 \eta}  \del_aH^{a\beta} \del_t u\del_{\beta} u. 
\endaligned
\ee
\ese 
In turn we arrive at a weighted energy estimate associated with the Euclidean-hyperboloidal foliation:  
\be \label{eq1-27-11-2022-M} 
\aligned
&\frac{d}{ds} \Eenergy_{g,\eta,c}(s,u)
+ 2\eta \int_{\Mscr_s}  \big( g^{\N ab} \delsN_au\delsN_bu + c^2u^2
\big)  \aleph'({ r-t}) \crochet^{2\eta-1}\ Jdx
\\
&= \int_{\Mscr_s} \Big(G_{g, \eta}[u] + \eta\crochet^{2\eta-1} \aleph'({ r-t}) \HN^{00}  |\del_t u|^2   \Big)\ Jdx 
+ \int_{\Mscr_s}  |\del_t u F | \, \crochet^{2 \eta} \, Jdx, 
\endaligned
\ee 
in which the latter integral is controlled by 
\be
\aligned
\int_{s_0}^{s_1} \Eenergy_{\eta,c}(s,u)^{1/2} \, \big(\| F \|_{L^2(\MH_s)}
+ \big\| s\zeta F \big\|_{L^2(\MM_s)}
+ \| s \, \crochet^\eta f\|_{L^2(\Mext_s)} \big) \, ds. 
\endaligned
\ee
 It will be convenient to also use the notation $\Fenergy_{\eta,c}(s,u) = \big( \Eenergy_{\eta,c}(s,u)\big)^{1/2}$ and similarly with the subscript $c$ omitted. 
 

\subsection{Functional inequalities on Euclidean-hyperboloidal slices}

\paragraph{Sobolev inequalities.} 

The proof of the statements in this section can be found in \cite[Part 1]{PLF-YM-main}. 

\begin{proposition}[Sup-norm Sobolev inequality. Hyperboloidal domain]
\label{prop:glol-Soin}
For any function defined on a hypersurface $\MH_s$, the following estimate holds (in which $t^2 = s^2+ |x|^2$): 
$$
\sup_{\MH_s} t^{3/2} \, |u(t,x)|
\lesssim 
\sum_{|J| \leq 2} \| L^J u\|_{L^2(\MH_s)}
\simeq 
\sum_{m=0,1,2} \| t^m (\slashed \del^\H)^m  u\|_{L^2(\MH_s)}.   
$$
\end{proposition}

\begin{proposition}[Weighted sup-norm Sobolev inequality. Euclidean-merging domain]
\label{pro204-11-2}
Fix an exponent {$\eta \geq 0$} and set $C(\eta) := 1 + \eta + \eta^2$.
For all sufficiently regular functions defined in $\Mscr_{[s_0,s_1]}$ with $2 \leq s_0 \leq s \leq s_1$, one has  
\begin{subequations}
\begin{equation} \label{ineq 2 sobolev}
\textstyle
r  \crochet^\eta |u(t,x)| 
\lesssim 
C(\eta) \sum_{|I| + |J| \leq 2} \|\crochet^\eta  \delsME{}^I\Omega^J u\|_{L^2(\MME_s)}, 
\qquad (t,x)\in \MME_s, 
\end{equation}
\begin{equation} \label{ineq 1 sobolev}
\textstyle
r \crochet^\eta |u(t,x)| 
\lesssim
C(\eta) \sum_{|I| + |J| \leq 2} \|\crochet^\eta  \delsE{}^I\Omega^J u\|_{L^2( \Mext_s)},
\qquad 
(t,x)\in \Mext_s. 
\end{equation}
\end{subequations}
\end{proposition}


\paragraph{Hardy inequalities.} Several functional inequalities will be used to control   undifferentiated functions from the energy functional. 

\begin{proposition}[Weighted Hardy inequality on the Euclidean-hyperboloidal foliation]
\label{lem1-hardy}
Fix some exponent {$\eta \geq 0$}. For any function $u$ defined in $\Mscr_{[s_0,s_1]}$ and sufficiently decaying at infinity, one has  
$$ 
\| r^{-1} \crochet^\eta u\|_{L^2(\Mscr_s)} 
\lesssim  
\| \crochet^\eta  \delsEH u\|_{L^2(\Mscr_s)}
\lesssim \Fenergy_\eta(s,u).
$$
\end{proposition}

\begin{proposition} [Hardy-type inequalities in the hyperboloidal domain]
For all functions $u$ defined on a hyperboloid $\Hcal_s$:  
\bse
\be \label{equa-Hardy1} 
\Big\|{u \over r} \Big\|_{L^2(\Hcal_s)}\lesssim \sum_a \|\delu_a u\|_{L^2(\Hcal_s)}.  
\ee
On the other hand, for all functions defined within a slab of the hyperboloidal foliation $[s_0, s_1]$  
\be \label{equa-Hardy2} 
\aligned
\Big\| {u \over s_1} \Big\|_{L^2(\Hcal_{s_1})} \lesssim \,
&  \textstyle
\Big\| {u \over s_0} \Big\|_{L^2(\Hcal_{s_0})} + \sum_{a}\|\delu_a u\|_{L^2(\Hcal_{s_1})}
 \\
 & \textstyle
 \quad + \sum_a \int_{s_0}^{s_1} \Big( 
\|\delu_a u\|_{L^2(\Hcal_s)} + \|(s/t)\del_a u\|_{L^2(\Hcal_s)}  \Big) \, {ds
\over s}.
\endaligned
\ee
\ese
\end{proposition}
 
We point out that while the proof of \eqref{equa-Hardy1} is analogue to the proof of the standard Hardy inequality, in contrast the proof of \eqref{equa-Hardy2} is more involved and is based on computing the divergence of the vector field 
\be
\chi\Big( {r\over t} \Big)^2 \, \Big(0, {t \, x^a  \over 1+r^2} \,  {u^2 \over s^2}  \Big), 
\qquad
\chi \Big({r\over t} \Big) 
= \begin{cases}
0, \, & \displaystyle 
r/t\leq 1/3,
\\
1, \, &   \displaystyle \hskip.cm r/t \geq 2/3,
\end{cases}
\ee
a smooth cut-off function being introduced here in order to removes a wedge containing the center of coordinates. 

 
\paragraph{Poincar\'e inequalities.} In our analysis the following functional inequalities will also be useful.  
 
\begin{proposition}[Poincar\'e-type inequalities in the Euclidean-merging domain] 
\label{propo-Poincare-ext}
\bse
Fix an exponent {$\eta =1/2 +\delta$} with $\delta > 0$. 
For any function $u$ defined in $\MME_s = \{(t,x)\in\Mscr_s\, / \, |x|\geq \rhoH(s)\}$, one has  
\begin{equation} \label{Poincare-trex-sans-zeta}
\|  \crochet^{-1 + \eta}  u\|_{L^2(\MME_s)}
\lesssim
\big( 1+ \delta^{-1} \big) \|  \crochet^{\eta} \delsME u\|_{L^2(\MME_s)}
+ 
\|  r^{-1} \crochet^{\eta} u\|_{L^2(\MME_s)}, 
\end{equation}
\begin{equation} \label{Poincare-trex-zeta}
\|\crochet^{-1 + \eta}\zeta u\|_{L^2(\MME_s)}
\lesssim
(1+\delta^{-1})\|\crochet^{\eta}\zeta \delsME u\|_{L^2(\MME_s)}
+ 
\| r^{-1}\crochet^{\eta}\zeta u\|_{L^2(\MME_s)}. 
\end{equation}
\ese
\end{proposition} 


\section{Wave-Klein-Gordon equations in Euclidean-Hyperboloidal foliations}
\label{section-wkg}

\subsection{Pointwise estimates for wave equations}

We now investigate the decay of solutions to wave equations and state $L^\infty$--$L^\infty$ estimates.  Given some data $f, u_0, u_1$ with sufficient regularity and decay (so that Kirchhoff's formula below makes sense), we consider the solution $u=u(t,x)$ to the initial value problem
\begin{equation} \label{eq7-28-12-2020}
\Box u = f,   \quad
u(1,x) = u_0(x),  \quad
\del_t u(1,x) = u_1(x), 
\qquad x \in \RR^3. 
\end{equation}
In order to separate the contributions from the initial data and from the source, we find it convenient to use the short-hand notation $u = \Box^{-1}[u_0,u_1,f]$, 
$\Box^{-1}_\init[u_0,u_1] := \Box^{-1}[u_0,u_1,0]$, 
and $\Box^{-1}_\source[f] := \Box^{-1}[0,0,f]$. 
While referring to \cite{PLF-YM-main} for the analysis of initial data, let us here consider the effect of a source, namely the operator $\Box^{-1}_{\source}$.  An earlier result was established in \cite{PLF-YM-two} where the source was supported in the interior of a light cone --- a restriction we overcome in the statement below. Throughout, $\Lambda_{t,x} := \big\{ (\tau,y) \big/ \, t-\tau = |x-y|, \, 1\leq  \tau\leq t \big\}$ is the truncated cone from a point $(t,x)$.

\begin{proposition}[Wave equation. Contribution from the source]
\label{Linfini wave}
Consider the wave operator  $\Box^{-1}_\source$ acting on a source function $f$ satisfying the decay conditions (with $\alpha_1, \alpha_2, \alpha_3$) 
\begin{equation} \label{eq1-27-12-2020}
|f(\tau,y)| \lesssim C_1 \, 
\tau^{\alpha_1}(\tau + |y| )^{\alpha_2} \big( 1 + | \tau - |y| | \big)^{\alpha_3}, 
\qquad 
(\tau,y) \in \Lambda_{t,x}. 
\end{equation}
Then the solution to the wave equation enjoys the following properties.

\vskip.15cm

\begin{subequations}

\noindent{\it  -- Case 0 (interior).} When the support of $f$ is contained in $\{r\leq t-1\}$, and $\alpha_1 = -2-\upsilon, \alpha_2 = 0, \alpha_3 = -1+\mu$ for $0<|\upsilon|\leq 1/2$, $0<\nu\leq 1/2$, one has 
\begin{equation}\label{eq1-29-11-2022-M}
|\Box^{-1}_\source[f](t,x)|
\lesssim
C_1
\begin{cases}
\frac{1}{\nu\mu}(t-r)^{\mu-\upsilon} t^{-1}, \qquad & 0< \nu\leq 1/2,
\\
\frac{1}{|\nu|\mu}(t-r)^{\mu} t^{-1 -\upsilon}, &-1/2\leq \nu < 0.
\end{cases}
\end{equation}

\noindent  {\it  -- Case 1 (typical).} When $\alpha_1 = -1+\upsilon$ and $\alpha_2 = -1-\nu$ and $\alpha_3=-1+\mu$ 
for some  
$\upsilon + \mu < \nu$ and $0<\mu,\nu,\upsilon\leq 1/2$,
one has 
\begin{equation} \label{eq5-24-12-2020}
|\Box^{-1}_\source[f](t,x)|
\lesssim
C_1 \, \big(\upsilon^{-1} + \mu^{-1} + |\mu-\nu|^{-1}\big) \, |\upsilon + \mu-\nu|^{-1} (t+r)^{-1}.
\end{equation} 

\vskip.15cm

\noindent{\it  -- Case 2 (sub-critical).} When $\alpha_1=0$ and $\alpha_2 = -2-\nu$ and $\alpha_3 = -1+\mu$ 
for some $0 < \nu, \mu \leq 1/2$, one has 
\begin{equation} \label{Linfini wave ineq}
|\Box^{-1}_\source[f](t,x)|
\lesssim C_1 \, 
\begin{cases} 
\mu^{-1} |\mu-\nu|^{-1} (t+r)^{-1}t^{\mu-\nu},\qquad  
&\mu>\nu,
\\
\mu^{-1} (t+r)^{-1} \ln (t+1),\quad  & \mu=\nu, 
\\
\mu^{-1} |\mu-\nu|^{-1} (t+r)^{-1}, & \mu<\nu.
\end{cases} 
\end{equation} 

\vskip.15cm

\noindent{\it  -- Case 3 (critical).} When $\alpha_1 = 0$ and $\alpha_2 = -2$ and $\alpha_3= -1-\mu$ for some 
$\mu\in (0,1/2)$, one has  
\begin{equation} \label{eq1-10-01-2021}
|\Box^{-1}_\source[f](t,x)| \lesssim C_1\,\mu^{-1} (t+r)^{-1}\Big(1 + \crochet^{-\mu}\ln\big(t / \crochet \big)\Big).
\end{equation} 

\vskip.15cm

\noindent{\it   -- Case 4 (super-critical).} When $\alpha_1 = 0$ and $\alpha_2 = -2+\nu$ and $\alpha_3 = -1-\mu$ for some 
$0<\nu< \mu < 1/2$, one has  
\begin{equation} \label{eq1-10-01-2021-case4}
|\Box^{-1}_\source[f](t,x)| \lesssim
C_1 \, \big(|\mu-\nu|^{-1} + \mu^{-1}\nu^{-1}\crochet^{-\mu}t^{\nu}\big) (t+r)^{-1}.
\end{equation} 
(In the last two cases, observe that $\crochet \equiv 1$ when $r\leq t-1$). 
\end{subequations}
\end{proposition}


\subsection{Pointwise estimates for Klein-Gordon equations}

The method below was introduced in Klainerman~\cite{Klainerman85}  
and the derivation below was proposed in LeFloch and Ma~\cite{PLF-YM-one}. 
Here, we present yet another version of the argument which 
 takes into account the contribution from the boundary (namely the light cone). 
 
We focus here on the hyperboloidal domain and rely on the decomposition 
\begin{equation}
g^{\alpha\beta}\del_{\alpha}\del_{\beta} = s^{-3/2}(t/s)^2\gu^{00}\Lcal^2(s^{3/2}\phi) + \RR_g[\phi],
\end{equation}
where $g^{\alpha\beta} = g_{\Mink}^{\alpha\beta} + H^{\alpha\beta}$ and 
the following field (which is nothing but the unit normal to the hyperboloids for the Minkowski metric) 
\be
\Lcal := (s/t)\del_t + (x^a/s)\delus_a = (t/s)\del_t + (x^a/s)\del_a,
\ee
while 
\be
\aligned
\RR_g[\phi] =& -(3/4)s^{-2}\phi - 3s^{-1} (x^a/s)\delus_a\phi - (x^ax^b/s^2)\delus_a\delus_b\phi - \sum_a\delus_a\delus_a\phi
\\
&-(t/s)^2\Hu^{00}\Big((3/4)s^{-2}\phi + \big(3+(r/t)^2\big)t^{-1}\del_t\phi + 3s^{-1}(x^a/s)\delus_a\phi\Big)
\\
&-(t/s)^2\Hu^{00}\Big((2x^a/t)\del_t\delus_a\phi + (x^ax^b/s^2)\delus_a\delus_b\phi\Big)
\\
& \quad +\Hu^{a0}\delus_a\del_t\phi + \Hu^{0a}\del_t\delus_a\phi + \Hu^{ab}\delus_a\delus_b\phi 
+ H^{\alpha\beta}\del_{\alpha}\big(\PsiH_{\beta}^{\beta'}\big)\delu_{\beta'}\phi.
\endaligned
\ee
For a Klein-Gordon equation 
$
\Boxt_g \phi - c^2\phi = f,
$
thanks to $\gu^{00} = -(s/t)^2 + \Hu^{00}$, under the assumption
$
(s/t)^2|\Hu^{00}|\leq 1/3,
$
the above identity leads us to
\begin{equation}
\Lcal^2(s^{3/2}\phi) + \frac{c^2}{1-(t/s)^2\Hu^{00}} (s^{3/2}\phi) = \frac{s^{3/2}(f + \RR_g)}{1-(t/s)^2\Hu^{00}}.
\end{equation}
For any function $\phi$ defined in $\MH_{[s_0,s_1]}$ and at each point $(t,x)\in \MH_{[s_0,s_1]}$, we use the notation $\Phi_{t,x}(\lambda) := \lambda^{3/2}\phi(\lambda t/s,\lambda x/s)$. Since $\Lcal(s^{3/2}\phi)|_{(\lambda t/s,\lambda x/s)} = \Phi_{t,x}'(\lambda)$, we find 
\begin{equation}\label{eq2-22-11-2022}
\Phi_{t,x}''(\lambda) + \frac{c^2}{1-\Hb_{t,x}}\Phi_{t,x} = \frac{\lambda^{3/2}}{1-\Hb_{t,x}}(f+\RR_g)|_{(\lambda t/s,\lambda x/s)},
\end{equation}
where  $\Hb_{t,x} = (t/s)^2\Hu^{00}|_{(\lambda t/s,\lambda x/s) }$. 

We observe that $(s/t)$ is constant along a given path $\gamma_{t,x}$, and we multiply the above equation by $(s/t)^{\eta}$ with $\eta\in\RR$. By an elementary ODE lemma, we then arrive at the following result. 

\begin{proposition}[Sharp decay of Klein-Gordon solutions in the hyperboloidal domain]
\label{prop1-23-11-2022-M}
Suppose that for all $(t,x)\in \MH_{[s_0,s_1]}$ and for all $\lambda_0\leq \lambda\leq s_1$, one has 
\begin{equation}\label{eq4-23-11-2022-M}
|\Hb_{t,x}|\leq 1/3,\qquad \int_{\lambda_0}^{s_1}|\Hb_{t,x}'(\lambda)|d\lambda\lesssim 1. 
\end{equation}
Then for any $\eta\in \RR$, any solution $\phi$ to the Klein-Gordon equation $\Boxt_g \phi - c^2\phi = f$ satisfies 
\begin{equation}\label{eq5-23-11-2022-M}
\aligned
& (s/t)^{\eta}s^{3/2} \, \Big(|\phi(t,x)|+ (s/t) \, |\del\phi(t,x)| \Big)
\\
& \lesssim (s/t)^{\eta}s^{1/2}|\phi|_1(t,x) 
+ \sup_{\MH_{s_0}\cup \Lscr_{[s_0,s]}}(s/t)^{\eta}\Big( t^{1/4}(t^{1/2}|\phi| + |\del\phi| + t \, |\delsN\phi|)\Big)
\\
& \quad + (s/t)^{\eta}\int_{\lambda_0}^s\lambda^{3/2}
\big( |f| + |\RR_g[\phi] \big)\big|_{(\lambda t/s,\lambda x/s)} d\lambda,
\endaligned
\end{equation}
in which 
\be
\lambda_0 = 
\begin{cases}
s_0,\quad &0\leq r/t\leq \frac{s_0^2-1}{s_0^2+1},
\\
\sqrt{\frac{t+r}{t-r}},&  \frac{s_0^2-1}{s_0^2+1}\leq r/t<1
\end{cases}
\ee
and 
\begin{equation}\label{eq11-23-11-2022-M}
\aligned
|\RR_g[\phi]|_{p,k}
& \lesssim s^{-2}|\phi|_{p+2} + (t/s)^2\sum_{p_1+p_2=p}|\Hu^{00}|_{p_1} \big(
s^{-2}|\phi|_{p_2+2} + t^{-1}|\del\phi|_{p_2+1} \big)
\\
& \quad + \sum_{p_1+p_2=p}|H|_{p_1} \big(
t^{-1}|\del\phi|_{p_2+1} + t^{-2}|\phi|_{p_2+2} \big).
\endaligned
\end{equation}
\end{proposition}

\begin{proof} Observe that the integral curve of $\Lcal$ is $\gamma_{t,x} = \{(\lambda t/s,\lambda x/s)\}$. When
$0\leq r/t\leq \frac{s_0^2-1}{s_0^2+1}$, the segment
$
\big\{( \lambda t/s,\lambda x/s)|s_0 \leq \lambda\leq s \big \}
$
is contained in $\MH_{[s_0,s]}$ and $\gamma_{t,x}$ meets $\del \MH_{[s_0,s]}$ at $(s_0/s)(t,x)\in\MH_{s_0}$. When $\frac{s_0^2-1}{s_0^2+1}\leq r/t<1$, the segment
$$
\Big\{(\lambda t/s,\lambda x/s)| \sqrt{\frac{t+r}{t-r}}\leq \lambda\leq s \Big\}
$$
is contained in $\MH_{[s_0,s_1]}$ and meets $\del\MH_{[s_0,s_1]}$ at the point $\frac{1}{t-r}(t,r)\in \MH_{[s_0,s_1]}\cap \MME_{[s_0,s_1]}$. Here we emphasize  that 
$
\sqrt{\frac{t+r}{t-r}}\simeq t/s.
$
We apply the multiplier $(s/t)^{2\eta}\Phi'_{t,x}$ to \eqref{eq2-22-11-2022} and obtain
$$
\frac{d}{d\lambda}\Big(|\Phi_{t,x,\eta}'|^2 + \frac{|\Phi_{t,x,\eta}|^2}{1-\Hb_{t,x}}\Big) - \frac{\Hb_{t,x}'}{(1-\Hb_{t,x})^2} |\Phi_{t,x,\eta}|^2 = \Phi_{t,x,\eta}'\frac{(s/t)^{\eta}(f + \RR_g)|_{(\lambda t/s,\lambda x/s)}}{1-\Hb_{t,x}}
$$
with $\Phi_{t,x,\eta} := (s/t)^{\eta}\Phi_{t,x}$. Then by Gronwall's inequality  applied on the interval $[\lambda_0,s]$, we find
$$
\aligned
|\Phi_{t,x,\eta}'(s)| + |\Phi_{t,x,\eta}(s)| & \lesssim  \big(|\Phi_{t,x,\eta}'(\lambda_0)| + |\Phi_{t,x,\eta}(\lambda_0)| \big)e^{\int_{\lambda_0}^s|\Hb_{t,x}(\lambda)|d\lambda}
\\
&\qquad + (s/t)^{\eta}\int_{\lambda_0}^s|f(\lambda t/s,\lambda x/s)|e^{\int_{\lambda}^s|\Hb_{t,x}(\tau)|d\tau}d\lambda.
\endaligned
$$
On the other hand, when $\lambda_0\geq s_0\geq 2$ we have 
$$
|\Phi_{t,x}'(\lambda)| + |\Phi_{t,x}(\lambda)|\simeq \lambda^{3/2}\big(|\phi(\lambda t/s,\lambda x/s)| + |\Lcal\phi(\lambda t/s,\lambda x/s)|\big)
$$
in which $\Lcal = (t/s)\del_t + (x^a/s)\del_a = (t/s)\del_t + (x^a/s)\del_a$. We also observe that
$$
\Lcal = (s/t)\del_t + s^{-1}(x^a/t)L_a,\quad \del_a = t^{-1}L_a - (x^a/t)\del_t,
$$
thus
$$
\aligned
& 
(s/t)^{\eta}\lambda^{3/2}\big(|\phi(\lambda t/s,\lambda x/s)| + (s/t) \, |\del \phi(\lambda t/s,\lambda x/s)|\big) 
\\
& \lesssim  (s/t)^{\eta}\lambda^{-1} |\phi|_1(\lambda t/s,\lambda x/s)
   +|\Phi_{t,x,\eta}'(\lambda)| + |\Phi_{t,x,\eta}(\lambda)|,
\endaligned
$$
and
$$
|\Phi_{t,x,\eta}'(\lambda)| + |\Phi_{t,x,\eta}(\lambda)|
\lesssim 
(s/t)^{\eta}\lambda^{3/2} \, \big(
|\phi|_1+(s/t) \, |\del\phi| + (t/s)|\delsN\phi|\big) \big|_{(\lambda t/s,\lambda x/s)}.
$$
This gives the desired result.
\end{proof}


We continue to rely on the linear structure of the Klein-Gordon equation. 
We use here the notation $\MMEnear_{[s_0,s_1]} := \MME_{[s_0,s_1]} \cap \big\{ t-1\leq r \leq 2t \big\}$, while the complement is denoted by $\Mfar_{[s_0,s_1]}$.  In the near light cone region we take advantage of the Klein-Gordon structure and control the mass term by the wave operator and a source term, while Sobolev decay is available in the far region. 

\begin{proposition}[Pointwise decay of Klein-Gordon fields]
\label{lem 1 d-KG-e}
Given any exponent $\eta \in (0,1)$, any  solution $v$ to $- \Box v + c^2 \, v = f$ defined in $\MME_{[s_0,s_1]}$ satisfies  
$$
c^2 \, |v|_{p,k} \lesssim 
\begin{cases}
r^{-2} \crochet^{1-\eta} \, \Fenergy_{\eta,c}^{\ME,p+4,k+4}(s,v) + |f|_{p,k}
&  \text{ in }\Mnear_{[s_0,s_1]},
\\
r^{-1-\eta} \, \Fenergy_{\eta,c}^{\ME,p+2,k+2}(s,v)\quad 
& \text{ in }\Mfar_{[s_0,s_1]}.
\end{cases}
$$
\end{proposition}


\section{Nonlinear stability of self-gravitating massive fields} 
\label{section-strategy}

\subsection{Nonlinear stability statement} 

\paragraph{Merging the Minkowski and Schwarzschild solutions.}

We present now our stability theory in the form established in \cite{PLF-YM-lambda1} (corresponding to $\lambda=1$ in~\cite{PLF-YM-main})
and refer the reader to \cite{PLF-YM-main} for the treatment of weaker spatial decay conditions. 
In wave coordinates the Schwarz\-schild metric $g_{\Sch}$ reads 
$$ 
\aligned
g_{\Sch,00} & =   - \frac{r-m}{r+m},
\quad g_{\Sch,0a} = 0,
\quad 
g_{\Sch, ab} = \frac{r+m}{r-m} \omega_a \omega_b + \frac{(r+m)^2}{r^2}(\delta_{ab} - \omega_a \omega_b), 
\endaligned
$$
with $\omega_a := x_a/r$. 
Let $\chi^\star(r)$ be (regular) cut-off function vanishing for all $r\leq 1/2$ and which is identically $1$ for all $r\geq 3/4$. Given a mass coefficient $m>0$, the reference metric of interest here is (by restricting attention 
to $t \geq 2$) for convenience in the discussion) 
\begin{equation}\label{equa-defineMS-new} 
g_\glue^\star = \gMink + \chi^\star (r) \, \chi^\star(r/(t-1)) (g_\Sch - \gMink), 
\qquad  
t \geq 2,
\end{equation} 
which coincides with $\gMink$ in the cone $\big\{ r/(t-1)< 1/2 \big\}$ and with $g_\Sch$ in the exterior $\big\{ r/(t-1)\geq 3/4 \big\}$ (containing the light cone).  
This metric satisfies the light-bending property in the sense that the coefficient 
\begin{equation}\label{equa-Sch-bending-new} 
r \, g_{\glue}^{\star}(\lbf,\lbf) = 4m + \Ocal(1/r) \qquad \text{ for the metric } g_\glue^\star, 
\end{equation}
is positive ---the  light cone direction being 
$
\lbf := \del_t - (x^a/r)\del_a.
$  


\paragraph{Class of initial data sets.}

The initial metric $g_0$ is assumed to be close to the Euclidean metric while the initial second fundamental form $k_0$ is small. We assume the following decomposition ($a, b=1,2,3$) 
\begin{equation}\label{equa-decomp-data-3} 
g_{0ab} = g^\star_{0ab} + u_{0ab} = \delta_{ab} + h^\star_{0ab} + u_{0ab}, 
\qquad 
k_{0 ab} = k^{\star}_{0ab} + l_{0ab},
\end{equation}
and we propose the following terminology. 

\begin{itemize}

\item The part $h^\star_{0}$ is referred to as the {\it  initial reference} and should be small in a (weighted, high-order) pointwise norm. 

\item The part $u_{0}$ is referred to as the {\it  initial perturbation} and should be small in  (weighted, high-order) energy norm. 

\end{itemize} 

An example of a such decomposition is provided by the construction in Lindblad and Rodnianski~\cite{LR1}, where the initial data is decomposed as the sum of a finite-energy perturbation plus an (asymptotically) Schwarzschild metric outside of a compact set (with sufficiently small and positive mass). In our theory~\cite{PLF-YM-main}, the two parts are treated differently. Indeed, $h_{0}^{\star}$ is the initial trace of $h^\star$ while $u_0$ is propagated.  


Let us fix some exponents $\kappa\in(1/2,1)$ and $\mu\in(3/4,1)$. For the metric perturbation and the matter field, we introduce the energy norms 
\begin{equation}\label{equa-norms} 
\aligned
\Fenergy^{\text{metric}}_{\kappa,N} (g_0, k_0)
:= &
\sum_{|I|\leq N}\big\|\la r\ra^{\kappa + |I|} \big( |\del_x^I\del_x u_0| + |\del_x^I u_1| \big) \big\|_{L^2(\RR^3)},
\\
\Fenergy^{\text{matter}}_{\mu,N} (\phi_0, \phi_1) 
:= &
\sum_{|I|\leq N}\big\|\la r\ra^{\mu + N} (|\del_x^I\del_x\phi_0| + |\del_x^I\phi_0| + |\del_x^I\phi_1|)\big\|_{L^2(\RR^3)}.
\endaligned
\end{equation}
Given an initial data set we decompose it according to \eqref{equa-decomp-data-3} and we introduce the {\sl linear development} denoted by $u_\init$ of the initial data set $(u_{0 \alpha\beta},u_{1 \alpha\beta})$, that is, we introduce the solution to the (free, linear) wave equation with this initial data. It can be checked that (using here that $\kappa>1/2$)
\begin{equation}\label{eq3-09-05-2021-new-L}
|u_{\init}|  \lesssim C_0\eps (t+r+1)^{-1}. 
\end{equation} 
Our main assumption beyond the smallness on the norms~\eqref{equa-norms} is the following
 {\sl light-bending condition:}
\begin{equation}\label{eq3'-27-05-2020-initial-new-Sch}
\inf_{\Mscr^\near_\ell} \big(4 m + r \, u_\init(\lbf,\lbf) \big) 
\geq m. 
\end{equation}
Here, the parameter $\ell \in (0,1/2]$ is fixed and we focus on the near-light cone domain 
$\Mscr^\near_{\ell} = \Big\{  t \geq 2, \quad t-1 \leq r \leq \frac{t}{1-\ell} \Big\}$. 

\bei 

\item either $\eps$ is small with respect to $m$,  
so that the contribution $m$ from the Schwarzschild metric dominates, 

\item or $u_\init(\lbf,\lbf) $ is non-negative  
(which can follow from positivity assumptions on the initial data and 
the fact that the fundamental solution to the wave equation is a non-negative measure),  

\item or yet a combination of the above two extreme examples, namely, the negative contribution of the perturbation is small with respect to the Schwarzschild mass. 

\eei

 
\paragraph{Main statement for the Einstein equations.}

We are in a position to state our main result. In fact, a slightly more general statement concerning perturbations of reference metrics with harmonic decay is actually established \cite{PLF-YM-lambda1} (while much weaker decay is proven to be sufficient for nonlinear stability in \cite{PLF-YM-main}. As explained earlier on in this text, similar results (in a rather different functional framework) was simultaneously and independently established by Ionescu and Pausader~\cite{IP3}. 

\begin{theorem}[Nonlinear stability of self-gravitating Klein-Gordon fields. Near Schwarzschild decay]
\label{theo-main-result} 
A constant $C_\star >0$ being fixed, the following result holds for all sufficiently small $\eps, m$
satisfying $\eps\leq C_\star m$. Consider the reference metric $g_\glue^\star$ defined in \eqref{equa-defineMS-new} by merging together the Minkowski  and Schwarzschild metrics. Consider constraint-satisfying initial data $(g_0,k_0,\phi_0,\phi_1)$, a large integer $N$, and  exponents $(\kappa,\mu,\eps)$  
satisfying  
\begin{equation}\label{eq2-10-04-2022-M-33}
\kappa\in(1/2,1), 
\qquad 
\mu\in(3/4,1),
\qquad 
\kappa\leq \mu.  
\end{equation}
Then provided the initial data satisfies the light-bending condition \eqref{eq3'-27-05-2020-initial-new-Sch} together with the smallness condition
\begin{equation}\label{eq1-10-04-2022-M}
\aligned 
\Ebf^{\text{metric}}_{\kappa,N}  (g_0, k_0)
+ \Ebf^{\text{matter}}_{\mu,N} (\phi_0, \phi_1) 
 \leq \eps, 
\qquad 
\endaligned
\end{equation}
the maximal globally hyperbolic Cauchy development of $(g_0,k_0,\phi_0,\phi_1)$ associated with the Einstein-massive field system
 is future causally geodesically complete, and asymptotically approaches Minkowski spacetime in all (timelike, null, spacetime) directions. Moreover, the component $g(\lbf, \lbf)$ has a harmonic decay and enjoys the light-bending condition, namely 
\begin{equation}\label{eq3-09-05-2021-new-L-g}
|g(\lbf, \lbf)| \lesssim {m+\eps \over t+r+1}, 
\qquad \quad
\inf_{\Mscr^\near_{\ell}} r \, g(\lbf,\lbf)  
\geq m/2.
\end{equation}
\end{theorem} 

In the rest of this paper, we will provide a full proof of nonlinear stability for a simplified model, which retains some of the main challenges arising with the Einstein-massive matter system; cf.~Theorem~\ref{theo-stable-model}. 


\subsection{Analysis of the Einstein equations in the Euclidean-hyperboloidal foliation} 

Next, having presented all of our technical tools, we turn our attention to the Einstein equations. The existence theory is established in wave gauge, and one of our tasks is to connect geometric components and 
PDEs components. 

\begin{itemize} 
\item {\it  Structure of Einstein's field equations.}
We decompose the Einstein-massive matter system in a form that is adapted to the Euclidean-hyper\-boloidal foliation 
and we analyze the nonlinear structure of these equations. In particular, we exploit the wave gauge conditions
and distinguish between different components of the metric.  
One major challenge comes from the fact that the nonlinearities arising in the Einstein equations
 do not obey the null condition. 
The wave gauge conditions play a central role in several instances, in the derivation of, both, energy and pointwise estimates.  
Cf. the statements in Lemmas~\ref{lem1-31-01-2021} and 
\ref{eq3 05-juillet-2019-hyper}. 

\item {\it  Consequences of the energy estimates.}
We proceed by postulating certain bootstrap assumptions which distinguish between low- and high-order derivatives of the metric and matter fields, and involve the translations, the boosts, and the spatial rotations. 
Our estimates involve the geometric weight denoted by $\zeta$, which allows us to 
distinguish between the interior and exterior domains of our foliation. 
Decay in space is incorporated by 
subtracting a reference metric and adding the weight $\crochet$ in terms of the distance from the light cone. A broad range of metric and matter exponents are allowed by our method. By applying the energy estimate, we derive directly several bounds and, in turn, we write 
direct consequences of the weighted Poincar\'e inequality in Proposition~\ref{propo-Poincare-ext}
(see for instance Proposition~\ref{eq3-15-05-2020}, below)
and 
of the generalized Sobolev inequality in Proposition~\ref{lem 2 d-e-I}. 

\item {\it  Commutator and Hessian estimates for the metric.}  
We then focus on the metric perturbation and establish estimates that are localized near the light cone as well as estimates
 away from it. We use various calculus rules enjoyed by our frame of vector fields 
and we uncover the boost-rotation hierarchy enjoyed by quasi-linear commutators, as 
stated earlier 
in Propositions~\ref{lm 2 dmpo-cmm-H} and \ref{prop1-12-02-2020-interior}.

\item {\it  Near-Schwarzschild decay of the null metric component.} 
A key contribution in our method is proving that a certain component of the metric, referred to as the
 null metric component, has a `near-Schwarzschild' decay; see~\eqref{eq3-09-05-2021-new-L-g}.
  In addition, a related argument allows us to prove a so-called light-bending condition: see \eqref{eq3-09-05-2021-new-L-g}. 
We find it useful to decompose the spacetime domain into two sub-domains, referred to as the ``bad'' and ``good'' regions:
in the bad region, which is a (thick) neighborhood of the light cone (covering points up to a distance $\sqrt{t}$) we 
integrate toward the light cone from the good region;
in the good region we apply Kirchhoff's formula and integrate the effect of the initial data, by taking the properties of the 
source terms into account. Here, we make use of the assumed decay of the reference metric and the contribution of the initial perturbation.  

\item {\it  Sharp decay for good metric components.}  
Our next task is to estimate the gradient and Hessian of the ``good metric components'' 
and derive suitably weighted pointwise estimate. 
In the wave equations satisfied by the metric perturbation, the  source-terms
 contain  the quasi-null terms $\Pbb$ which may not enjoy integrable decay.
By virtue of the tensorial structure, the quasi-null terms in the evolution equations of   
the good components of the metric are actually null terms and, consequently, enjoy sufficient decay. This allows us to uncover a hierarchy between the Einstein equations. 

\item {\it  Pointwise estimate for metric components at low order.}  
We next control general components of the metric at low order of differentiation,  
and we derive a near-Schwarzschild decay which is essential in order to deal with massive matter fields.
On the other hand, for massless fields a weaker estimate would be sufficient to close the bootstrap argument.   

\item {\it   Improved energy estimates.} 
In turn we can close the bootstrap argument by establishing improved energy estimates at the highest-order of differentiation, both, first for general metric components  
and then for the Klein-Gordon field. The boost-rotation hierarchy made evident in our earlier estimates is also here 
the key ingredient of this final step of the prof. For this argument applied to the model, cf.~Section~\ref{section---67}. 

\item {\it  Asymptotically hyperboloidal domain.} 
Estimates are required also within the asymptotically hyperboloidal domain. 
For this region, a proof of global existence was given first in our monograph \cite{PLF-YM-two} in which the emphasis was on spacetime coincide exactly with the Schwarzschild spacetime outside a (large, say) light cone. In the light cone region, we also 
investigated the structure of the nonlinearities of the Einstein equations coupled to a Klein-Gordon fields, and uncovered 
 the boost-rotation hierarchy, as well. 

\end{itemize}


\subsection{Null and quasi-null structures in the Euclidean-hyperboloidal foliation}

The Einstein equations can be decomposed in the frames that are relevant in the Euclidean--hyperboloidal foliation framework. In the global coordinate chart $(x^\alpha) =( t, x^a)$, we introduce   
$\Gamma^\gamma := 
g^{\alpha \beta} \Gamma_{\alpha \beta}^\gamma$ and $ \Gamma_\alpha := 
g_{\alpha \beta} \Gamma^\beta$
determined from the corresponding Christoffel symbols. 
\label{equa-the-system}
The Ricci curvature depends upon (up to) second-order derivatives of the metric $g$ and specifically (\cite{PLF-YM-two}) 
\begin{equation} \label{equa:sec8-01} 
2 \, R_{\alpha\beta}
= - g^{\mu\nu} \del_\mu \del_\nu g_{\alpha\beta} 
+  \Fbb_{\alpha\beta}(g,g;\del g, \del g) 
+ \big(\del_\alpha\Gamma_\beta + \del_\beta\Gamma_\alpha\big) 
+  W_{\alpha\beta},
\end{equation}
where $W_{\alpha\beta} := g^{\delta \delta'} \del_{\delta} g_{\alpha \beta} \Gamma_{\delta'} - \Gamma_\alpha \Gamma_\beta$ and $\Fbb_{\alpha\beta} = \Pbb_{\alpha\beta} + \Qbb_{\alpha\beta}$ which involves 
the
{\it   
quasi-null quadratic forms} 
\begin{equation} \label{equa:sec8-03} 
\Pbb_{\alpha\beta} (g,g;\del g, \del g) 
:= - \frac{1}{2} g^{\mu\mu'} g^{\nu\nu'} \del_\alpha g_{\mu\nu} \del_\beta g_{\mu'\nu'} + \frac{1}{4} g^{\mu\mu'} g^{\nu\nu'} \del_\alpha g_{\mu\mu'} \del_\beta g_{\nu\nu'}
\end{equation}
and {\it  null quadratic forms}
\begin{equation}
\aligned
&\Qbb_{\alpha\beta}(g,g;\del g, \del g) 
 := 
g^{\mu\mu'} g^{\nu\nu'} \del_\mu g_{\alpha\nu} \del_{\mu'} g_{\beta\nu'}
- g^{\mu\mu'} g^{\nu\nu'} \big(\del_\mu g_{\alpha\nu'} \del_\nu g_{\beta\mu'} - \del_\mu g_{\beta\mu'} \del_\nu g_{\alpha\nu'} \big)
\\
&  + g^{\mu\mu'} g^{\nu\nu'} \big(\del_\alpha g_{\mu\nu} \del_{\nu'} g_{\mu'\beta} - \del_\alpha g_{\mu'\beta} \del_{\nu'} g_{\mu\nu} \big)
+ \frac{1}{2} g^{\mu\mu'} g^{\nu\nu'} \big(\del_\alpha g_{\mu\beta} \del_{\mu'} g_{\nu\nu'} - \del_\alpha g_{\nu\nu'} \del_{\mu'} g_{\mu\beta} \big)
\\
&  + g^{\mu\mu'} g^{\nu\nu'} \big(\del_\beta g_{\mu\nu} \del_{\nu'} g_{\mu'\alpha} - \del_\beta g_{\mu'\alpha} \del_{\nu'} g_{\mu\nu} \big)
+ \frac{1}{2} g^{\mu\mu'} g^{\nu\nu'} \big(\del_\beta g_{\mu\alpha} \del_{\mu'} g_{\nu\nu'} - \del_\beta g_{\nu\nu'} \del_{\mu'} g_{\mu\alpha} \big). 
\endaligned
\end{equation}


We  focus on the quadratic terms in the perturbation, namely
\be
\Pbb^{\star}_{\alpha\beta}[u,u] := \Pbb_{\alpha\beta}(g^{\star},g^{\star}; \del u,\del u),\quad
\Qbb^{\star}_{\alpha\beta}[u,u] := \Qbb_{\alpha\beta}(g^{\star},g^{\star}; \del u, \del u).
\ee
Since these expressions are quadratic in $(g,g)$ as well as in $(\del g, \del g)$, we apply a polarization argument and define the corresponding symmetric bilinear forms.  We consider  
$|\Qbb^{\star}[u, v]|_{p,k} := \max_{\alpha,\beta} |\Qbb_{\alpha\beta}^\star[u, v] |_{p, k}$.

\begin{lemma}[Null interaction terms at arbitrary order] 
\label{Null-Euclidean bilinear}
Null forms are controlled by good derivatives and a contribution depending upon the reference metric and, specifically,
n the Euclidean-merging domain one has
\begin{equation} \label{equa-new-Qzero}
\aligned
|\Qbb^{\star}[u, v]|_{p,k} 
& \lesssim  \sum_{p_1+p_2 = p\atop k_1+k_2=k} 
\Big( |\del u|_{p_1, k_1} |\delsN v |_{p_2, k_2}  + |\del v |_{p_1, k_1} |\delsN u|_{p_2, k_2}  \Big) 
\\
& \textstyle
\quad + | h^\star |_p  
\sum_{p_1+p_2=p\atop k_1+k_2=k} |\del u|_{p_1, k_1} |\del v |_{p_2, k_2}, 
\endaligned
\end{equation} 
while, in the hyperboloidal domain $\Mscr^\H$,  
\begin{equation}
|\Qbb^\star[u] |_p
\lesssim  \sum_{p_1+p_2 = p} |\del u|_{p_1} \big( |\delsH u|_{p_2} + (s/t)^2 |\del u|_{p_2} \big)
+ | h^\star |_p\sum_{p_1+p_2=p} |\del u|_{p_1} |\del u|_{p_2}.
\end{equation}
\end{lemma}

Next, in the Euclidean-merging domain let us introduce 
\begin{equation}\label{eq8-04-10-2022} 
w[v]_{\gamma} := (\gMink + v)^{\alpha\beta} \del_{\alpha} v_{\beta\gamma} - \frac{1}{2} (\gMink + v)^{\alpha\beta} \del_\gamma v_{\alpha\beta},
\end{equation}
\begin{equation} \label{eq2-04-12-2020-v}
\aligned
\mathbb{W}^{\ME}_{p,k}[v]
&:=    
{ \sum_\gamma |w[v]_\gamma|_{p,k} 
}
+ \sum_{p_1+p_2 = p\atop k_1+k_2=k} |\del  v |_{p_1,k_1} | v |_{p_2,k_2}. 
\endaligned
\end{equation}

\begin{lemma}[Quasi-null interaction terms at arbitrary order. Euclidean-merging domain] 
\label{lem1-31-01-2021}
In the Euclidean-merging domain $\MME$, under the smallness condition $| h^\star |_p + |u|_{[p/2]} \ll 1$ and 
$|v|_{[p/2]} \ll 1$ one has 
$$
\aligned
|\slashed{\Pbb}^{\star\N}[u, v] |_{p,k}
&
\lesssim  \hskip-.2cm 
\sum_{p_1+p_2=p\atop k_1+k_2=k}\big(  |\del u|_{p_1,k_1} |\delts v |_{p_2,k_2} +  |\del \phi |_{p_1} |\delts u|_{p_2} \big) 
+  \hskip-.4cm  \sum_{p_1+p_2+p_3=p}  \hskip-.3cm  | h^\star |_{p_3} |\del u|_{p_1} |\del \phi |_{p_2},
\\
|\Pbb_{00}^{\star \Ncal}[u, v] |_{p,k}
& \lesssim \sum_{p_1+p_2=p\atop k_1+k_2=k} |\del \uts|_{p_1,k_1} |\del \slashed v^\Ncal |_{p_2,k_2} 
+ \hskip-.3cm  {\sum_{p_1+p_2=p} |\delts u|_{p_1} |\del \phi |_{p_2}} 
+  \hskip-.3cm  \sum_{p_1+p_2=p}  \hskip-.3cm   |\delts v |_{p_1} |\del u |_{p_2} 
\\
& \quad  
+  \hskip-.3cm  \sum_{p_1+p_2=p\atop k_1 + k_2 = k} \mathbb{W}^{\EM}_{p_1,k_1}[v] \, |\del u |_{p_2,k_2}
+ \hskip-.3cm  \sum_{p_1+p_2=p}\hskip-.3cm \SbbME_{p_1}[u] |\del \phi |_{p_2} 
+ \hskip-.5cm \sum_{p_1+p_2+p_3=p} \hskip-.5cm  | h^\star |_{p_3} |\del u|_{p_1} |\del \phi |_{p_2}.
\endaligned
$$
\end{lemma} 
 

Finally, in the hyperboloidal domain we introduce 
\begin{equation}\label{eq2-04-12-2020-hyper}
\aligned
\Sbb_p^\H[u] 
& := t^{-1} |u|_{p_1} + \big(|\del h^\star |_{p_1} + t^{-1} | h^\star |_{p_1} \big)
+ \sum_{p_1+p_2=p_1} \big(|\del  u|_{p_1} |u|_{p_2} + |u|_{p_1} |u|_{p_2} \big)
\\
& \quad +   \sum_{p_1+p_2=p_1} \big(| h^\star |_{p_1} |\del u|_{p_2} + |u|_{p_1} |\del h^\star |_{p_2} + | h^\star |_{p_1} |\del h^\star |_{p_2} \big)
\\
& \quad
+  \sum_{p_1+p_2=p_1} 
\big(| h^\star |_{p_1} | u|_{p_2} + |u|_{p_1} | h^\star |_{p_2} + | h^\star |_{p_1} | h^\star |_{p_2} \big). 
\endaligned
\end{equation}


\begin{lemma}[Quasi-null interaction terms at arbitrary order. Hyperboloidal domain] 
\label{eq3 05-juillet-2019-hyper}
In the hyperboloidal domain $\Mcal^\Hcal$ and under the smallness condition $| h^\star |_p + |u|_{[p/2]} \ll 1$,  the quasi-null terms satisfy 
$$
\aligned
|\slashed{\Pbb}^{\star\H} [u] |_p
& \lesssim \sum_{p_1+p_2=p} |\del u|_{p_1} \big(  |\delsH u|_{p_2} + (s/t)^2 |\del u|_{p_2} \big) 
+ \sum_{p_1+p_2+p_3=p} | h^\star |_{p_3} |\del u|_{p_1} |\del u|_{p_2},
\\
|\Pbb_{00}^{\star\H} [u] |_{p,k}
& \lesssim     \hskip-.3cm  \sum_{p_1+p_2=p}   \hskip-.3cm  \Big(|\del \usH |_{p_1} |\del \usH |_{p_2} 
+ \big( |\delsH u|_{p_1} + (s/t)^2 |\del u|_{p_1} \big)  |\del u|_{p_2}
\Big) 
\\
& \qquad 
+   \hskip-.3cm 
\sum_{p_1+p_2=p}  \hskip-.3cm  | \Sbb_p^\H[u] |_{p_1} |\del u|_{p_2}  
+  \hskip-.3cm  \sum_{p_1+p_2+p_3=p}   \hskip-.3cm  | h^\star |_{p_3} |\del u|_{p_1} |\del u|_{p_2}.
\endaligned
$$ 
\end{lemma} 


\section{Nonlinear stability for the wave-Klein-Gordon model} 
\label{section-555}

We work with the proposed Euclidean--hyperboloidal foliation  $\bigcup_{s \geq 1} \Mscr_s$ 
of Minkowski spacetime.
For any integer $N \geq 1$ and any values $s \geq 1$ of the foliation parameter,
we consider the following weighted energy functional for the wave component $u$ on the hypersurface $\Mscr_s$: 
\begin{equation}
\Eenergy^N_\kappa(s,u) := \sum_{\ord(Z)\leq N} \Eenergy_{\kappa}(s, Z u), 
\qquad \Fenergy_{\kappa}^N (s,u)
:= 
\big( \Eenergy^N_\kappa(s,u) \big)^{1/2},
\end{equation}
while for the Klein-Gordon field with mass coefficient $c>0$ we set 
\begin{equation}
\Eenergy^N_{c,\kappa}(s, \phi) := \sum_{\ord(Z) \leq N} \Eenergy_{c, \kappa}(s, Z \phi),\qquad
\Fenergy^N_{c,\kappa}(s, \phi) 
:= \big( \Eenergy^N_{c,\kappa}(s, \phi) \big)^{1/2}.
\end{equation}
The summations above are over all ordered admissible operators of order $\leq N$. Later on, we will also need the notation
 $\Fenergy_\kappa^{p, k}(s,u)$ for the energy defined by restricting the summation to fields with $\ord(Z) \leq p$ and $\rank(Z) \leq k$. We also write $\Fenergy_\kappa^{\ME,p, k}(s,u)$  and $\Fenergy^{\H,p, k}(s,u)$ 
 when the integrals defining the energy is restricted to the domains $\MME$ or $\MH$, respectively. (The subscript $\kappa$ is irrelevant and omitted for the energy in the hyperboloidal domain.)

We now state our main result for
 the model.

\begin{theorem}[Global existence theory for the wave-Klein-Gordon model]
\label{theo-stable-model}
Consider the nonlinear wave-Klein-Gordon model \eqref{eq 1 model} with given real constants
$P^{\alpha\beta}, R, H^{\alpha\beta}$, and a mass coefficient $c>0$. 
For $\kappa > 3/4$ and any integer $N \geq 11$, there exists a sufficiently small $\eps>0$ such that the  initial value problem associated with the system \eqref{eq 1 model} admits a global-in-time solution $(u,v)$, 
 provided the data set on the initial hypersurface  has sufficiently small energy in the sense that  (for a fixed $C_0>0$) 
\begin{equation}\label{eq h initi}
\Fenergy^N_\kappa(s_0,u) + \Fenergy^N_{c, \kappa}(s_0,\phi)
\leq C_0 \, \eps,
\end{equation}
\begin{equation}\label{eq l initi}
\Fenergy_\kappa^{N-4}(s_0,u) + \Fenergy^{N-4}_{c, \kappa}(s_0,\phi)
\leq C_0 \, \eps
\end{equation}
and, for every admissible field $Z$ with $\ord(Z) \leq N-4$, the solution $w$ to the following 
free wave problem 
\be
\aligned
\Box w = 0, \quad 
\qquad
 w |_{\Mscr_{s_0}} & = w_0, \quad\qquad  & \del_t w |_{\Mscr_{s_0}} & = w_1,  
\\
  w_0 & = Z u  |_{\Mscr_{s_0}}, \quad      & w_1 & = \del_t Z u  |_{\Mscr_{s_0}}, 
\endaligned
\ee
satisfies the decay property 
\begin{equation}\label{eq172-bis-22}
|w(t,x)| \lesssim \eps \, (r+t)^{-1}. 
\end{equation}
\end{theorem}

On the other hand, similarly as we did for the Einstein equations, we could also introduce a notion of reference solution and establish the above existence theory under milder decay conditions.  On the other hand, we point out that \eqref{eq172-bis-22}
is satisfied not 
only by compactly supported initial data, but also by a class of non-compact initial data described in \cite[Section 10]{PLF-YM-main}. 
To proceed, we fix some exponents 
\begin{equation}\label{eq:choiceexpo}
3/4<\kappa <1, 
\qquad
0 < \delta\ll \kappa-3/4, 
\end{equation}
where the exponent $\delta$ is used to allow a (mild) growth of our energy norms. 
(For instance, it is  possible to take $\delta = (\kappa-3/4)/20$.) 
Our {\sl bootstrap energy assumptions} are stated on a time interval $s \in [s_0, s_1]$ and are based 
on a sufficient large constant $C_1>0$ (in comparison to $C_0$ in \eqref{eq h initi} and \eqref{eq l initi}) which 
will be chosen later on: 
\begin{equation}\label{eq h bootstrap}
s^{-\delta} \, \Fenergy^N_\kappa(s,u)
+ s^{-1/2-\delta} \, \Fenergy^N_{c, \kappa}(s, \phi)
\leq C_1\eps,
\end{equation}
\begin{equation}\label{eq l bootstrap}
\Fenergy_\kappa^{N-4}(s,u) + s^{-\delta} \, \Fenergy^{N-4}_{c, \kappa}(s, \phi)
\leq C_1\eps.
\end{equation} 
We emphasize that, in \eqref{eq h bootstrap}, the high-order energy norm of the wave component may grow at the mild rate $s^\delta$, 
while the high-order energy norm of the Klein-Gordon component may grow at the rate $s^{1/2+\delta}$. On the other hand, in \eqref{eq l bootstrap} a uniform control is required on the low-order energy norm of the wave component, 
while the Klein-Gordon component may suffer a mild growth~$s^\delta$. 


\subsection{High-order estimates} 

Functional inequalities together with calculus rules 
 allow us to derive estimates at arbitrary orders and for instance we arrive at the following statement. 

\begin{proposition}[Sobolev decay for wave fields in the Euclidean-merging domain] 
\label{lem 2 d-e-I}
For all $\eta \in [0,1)$ and all functions $u$, one has (for $k \leq p$)
\begin{subequations}
\begin{equation} \label{eq 1 lem 2 d-e-I}
\aligned
& 
\big\| r  \, \crochet^\eta \, |\del u|_{p,k} \big\|_{L^\infty(\MME_s)}  
+ \big\|  r^{1+ \eta} \, | \delsN  u |_{p,k} \big\|_{L^\infty(\MME_s)} 
\\
& \lesssim (1-\eta)^{-1} \, \Fenergy_\eta^{\ME,p+3, k+3}(s,u)
\endaligned
\end{equation} 
and, for $\eta \in (1/2, 1)$,  
\begin{equation} \label{eq 1 lem 2 d-e-I-facile}
\| r \, \crochet^{-1+\eta} |u|_{N-2}\|_{L^\infty(\MME_s)} 
\lesssim (2\eta-1)^{-1} 
\, \Fenergy_\eta^{\ME,N}(s,u) + \Fenergy_{\eta}^{0}(s,u). 
\end{equation}
\end{subequations}
\end{proposition}

 
\begin{proposition}[Sobolev decay for wave fields in the hyperboloidal domain]\label{prop1-11-22-2022-M}
For all function $u$ defined in $\MH_{[s_0,s_1]}$, one has (for $k\leq p$)
\begin{equation}\label{eq1-22-11-2022-M}
\|t^{3/2}(s/t) \, |\del u|_{p,k}\|_{L^{\infty}(\MH_s)} + \|t^{3/2}|\delus u|_{p,k}\|_{L^{\infty}(\MH_s)}
\lesssim \Fenergy^{\H,p+2,k+2}(s,u),
\end{equation}
\begin{equation}\label{eq2-22-11-2022-M}
\aligned
& \|t^{3/2}(s/t) \, |\del \phi|_{p,k}\|_{L^{\infty}(\MH_s)} + 
\|t^{3/2}|\delus \phi|_{p,k}\|_{L^{\infty}(\MH_s)} 
+ c\|t^{3/2}|\phi|_{p,k}\|_{L^{\infty}(\MH_s)}
\\
& \lesssim \Fenergy_c^{\H,p+2,k+2}(s,u).
\endaligned
\end{equation}
\end{proposition}

\begin{proposition}[Hardy-Poincar\'e inequality for high-order derivatives] 
\label{eq3-15-05-2020}
For any $\eta>1/2$ and any function $u$ defined in $\Mscr_{[s_0,s_1]}$ and all $s \in [s_0, s_1]$ one has 
$$
\| \crochet^{-1 + \eta} |u|_{p,k}\|_{L^2(\MME_s)} 
\lesssim \big(1+ (2 \eta-1)^{-1} \big) \, \Fenergy_\eta^{\ME,p,k}(s,u) + \Fenergy_{\eta}^{0}(s,u). 
$$
\end{proposition} 


\section{Proof of stability}
\label{section---666}

\subsection{Direct consequences in the Euclidean-merging domain} 

\paragraph{Hardy-Poincar\'e inequality.}

Throughout the following derivation, the notation $A\lesssim B$ stands for $A\leq C_{N,\delta}B$ where
 $C_{N,\delta}$ a constant depending upon $N,\delta$. In view of the bootstrap conditions \eqref{eq h bootstrap}-\eqref{eq l bootstrap} and Proposition~\ref{eq3-15-05-2020}, we have 
\begin{equation}\label{eq1-19-11-2022-M}
\|\crochet^{-1+\kappa}|u|_p\|_{L^2(\MME_s)}\lesssim
\begin{cases}
C_1\eps,\quad & p  =  N-4,
\\
C_1\eps s^{\delta}, & p  =   N.
\end{cases}
\end{equation}


\paragraph{Sobolev decay.}

In view of \eqref{eq h bootstrap}-\eqref{eq l bootstrap} together with \eqref{eq1-19-11-2022-M}, we can apply 
Proposition \ref{lem 2 d-e-I} to both the wave field $u$ and the Klein-Gordon field $\phi$, 
and we obtain
\begin{equation}\label{eq2-19-11-2022-M}
\la r\ra \crochet^{\kappa}|\del u|_p + \la r\ra^{1+\kappa}|\delsN u|_p\lesssim 
\begin{cases}
C_1\vep,\quad & p  =  N-7,
\\
C_1\vep s^{\delta},\quad & p  =   N-3,
\end{cases}
\quad \text{ in } \MME_{[s_0,s_1]}, 
\end{equation}
\begin{equation}\label{eq3-19-11-2022-M}
\la r\ra\crochet^{-1+\kappa}|u|_p\lesssim
\begin{cases}
C_1\vep,\quad & p  =   N-6,
\\
C_1\vep s^{\delta}\quad & p  =   N-2, 
\end{cases}
\quad \text{ in } \MME_{[s_0,s_1]}, 
\end{equation}
\begin{equation}\label{eq4-19-11-2022-M}
\la r\ra\crochet^{\kappa}|\del \phi|_p + \la r\ra^{1+\kappa}|\delsN \phi|_p + \la r\ra\crochet^{\kappa}|\phi|_{p+1}\lesssim
\begin{cases}
C_1\vep s^{\delta}, \,  & p  =   N-7,
\\
C_1\vep s^{1/2+\delta}, \,  & p  =  N-3, 
\end{cases}
\text{ in } \MME_{[s_0,s_1]}.
\end{equation}


\paragraph{Bounds on the light cone.}

Along the light cone $\Lscr_{[s_0,s]} = \MH_{[s_0,s_1]}\cap \MME_{[s_0,s_1]}$, we have $\delus_a = \delsN_a - (x^a/r)t^{-1}\del_t$. By restricting \eqref{eq2-19-11-2022-M} to $\Lscr_{[s_0,s]}$ and by observing that $t\simeq s^2$, we find 
\begin{equation}
t \, |\del u|_p + t^{1+\kappa}|\delus u|_p\lesssim
\begin{cases}
C_1\eps,\quad & p  =   N-7,
\\
C_1\eps \, t^{\delta/2}, & p  =  N-3,
\end{cases}
\qquad \text{on } \Lscr_{[s_0,s]}, 
\end{equation}
and 
\begin{equation}\label{eq1-23-11-2022-M}
t \, |\del \phi|_p + t^{1+\kappa}|\delus \phi|_p + t \, |\phi|_{p+1}\lesssim
\begin{cases}
C_1\eps \, t^{\delta/2},\quad & p  =   N-7,
\\
C_1\eps \, t^{1/4+\delta/2}, & p  =   N-3,
\end{cases}
\qquad \text{on } \Lscr_{[s_0,s]}.
\end{equation}
In the same manner, from \eqref{eq3-19-11-2022-M} we deduce 
\begin{equation}\label{eq3-22-11-2022-M}
t \, |u|_p\lesssim
\begin{cases}
C_1\eps,\quad & p  =  N-6,
\\
C_1\eps \, t^{\delta/2},\quad & p  =  N-2,
\end{cases} 
\qquad \text{on } \Lscr_{[s_0,s]}
\end{equation}
and, in particular, this leads us to the bound (since $\delus_a = t^{-1}L_a$):
\begin{equation}\label{eq12-23-11-2022-M}
t^2|\delus u|_{p-1}\lesssim
\begin{cases}
C_1\eps,\quad & p   =   N-6,
\\
C_1\eps \, t^{\delta/2},\quad & p   =   N-2,
\end{cases} 
\qquad \text{on } \Lscr_{[s_0,s]}.
\end{equation}


\subsection{Direct consequences in the hyperboloidal domain} 

\paragraph{Sobolev decay.}

In view of the Sobolev decay stated in Proposition~\ref{prop1-11-22-2022-M} in combination with the bootstrap bounds \eqref{eq h bootstrap}-\eqref{eq l bootstrap}, we have 
\begin{equation}\label{eq7-23-11-2022-M}
t^{3/2}(s/t) \, |\del u|_p + t^{3/2}|\delus u|_p 
\lesssim 
\begin{cases}
C_1\eps,\quad & p = N-6,
\\
C_1\eps s^{\delta}, & p = N-2.
\end{cases}
\quad \text{ in } \MH_{[s_0,s_1]}, 
\end{equation}
\begin{equation}\label{eq2-24-11-2022-M}
t^{3/2}(s/t) \, |\del \phi|_p + t^{3/2}|\delus \phi|_p + t^{3/2}|\phi|_p 
\lesssim 
\begin{cases}
C_1\eps s^{\delta},\quad & p = N-6,
\\
C_1\eps s^{1/2+\delta}, & p = N-2, 
\end{cases}
\quad \text{ in } \MH_{[s_0,s_1]}.
\end{equation}
Furthermore, this implies that
\begin{equation}\label{eq6-23-11-2022-M}
(s/t)^{-2+\delta}s^{1/2}|\phi|_{N-3}\lesssim C_1\eps
\quad \text{ in } \MH_{[s_0,s_1]} 
\end{equation}
and, in particular,  
\be
\aligned
|\del_r (Z\delus_a u)| 
& \lesssim  
\begin{cases}
C_1\eps \, t^{-3/2}s^{-1},\quad & \ord Z\leq N-7, 
\\
C_1\eps \, t^{-3/2}s^{-1+\delta}, & \ord Z\leq N-3,
\end{cases}
\\
& \simeq 
\begin{cases}
C_1\eps \, t^{-2}(t-r)^{-1/2},\quad & \ord Z\leq N-7, 
\\
C_1\eps \, t^{-2+\delta/2}(t-r)^{-1/2+\delta/2}, & \ord Z\leq N-3,
\end{cases}
\quad \text{ in } \MH_{[s_0,s_1]}.
\endaligned
\ee


\paragraph{Integration from the light cone.}

Recalling \eqref{eq12-23-11-2022-M} and performing an integration ---along a constant-time slice--- from 
the light cone $\Lscr_{[s_0,s]} = \MH_{[s_0,s_1]}\cap \MME_{[s_0,s_1]}$
 towards the center, we deduce that
\begin{equation}\label{eq8-23-11-2022-M}
|\delus u|_p\lesssim
\begin{cases}
C_1\eps \, t^{-3/2}(s/t),\quad & p { =  N-7,}
\\
C_1\eps \, t^{-3/2}(s/t)s^{\delta}, & p { =  N-3,}
\end{cases}
\quad \text{ in } \MH_{[s_0,s_1]}.
\end{equation}
In the same manner, we can integrate $\del_r u = (x^a/r)\del_au$ from $\Lscr_{[s_0,s_1]}$ towards the center on a constant-time  slice and, in view of \eqref{eq3-22-11-2022-M} and \eqref{eq7-23-11-2022-M}, we find
\begin{equation}\label{eq4-24-11-2022-M}
|u|_p\lesssim
\begin{cases}
C_1\eps \, t^{-1}(t-r)^{1/2},\quad & p = N-6,
\\
 C_1\eps \, t^{-1+\delta/2}(t-r)^{1/2+\delta/2},  & p =  N-2, 
\end{cases}
\quad \text{ in } \MH_{[s_0,s_1]}.
\end{equation}


\paragraph{Further estimate.}

Thanks to \eqref{eq7-23-11-2022-M}, \eqref{eq8-23-11-2022-M} and the relation 
$\Hb_{t,x}'(\lambda) = (t/s)^2\Hu^{00}\Lcal u$,  
we have
\begin{equation}\label{eq9-23-11-2022-M}
|\Hb_{t,x}'(\lambda)|\lesssim C_1\eps (s/t)^{-1/2}s^{-3/2}
\quad \text{ in } \MH_{[s_0,s_1]} 
\end{equation}
(provided $N\geq 7$). This guarantees the assumption \eqref{eq4-23-11-2022-M} in Proposition~\ref{prop1-23-11-2022-M}
(by observing that $(t/s)\lesssim \lambda_0$) and this observation will be required in Section~\ref{section----64}, below. 


\subsection{Commutators and source terms in the Euclidean-merging domain} 

\paragraph{Pointwise decay for the Klein-Gordon component.}

By a direct application of the bounds \eqref{eq3-19-11-2022-M}-\eqref{eq4-19-11-2022-M} we have 
$$
\aligned
|H^{\alpha\beta}u\del_{\alpha}\del_{\beta}\phi|_{p-4} & \lesssim  (C_1\eps)^2 r^{-2}\crochet^{1-2\kappa}
\begin{cases}
s^{1/2+2\delta},\quad & p  =  N, 
\\
s^{2\delta},\quad & p  =   N-4,
\end{cases}
\\
& \lesssim  C_1\eps \la r\ra^{-2}\crochet^{1-\kappa}\begin{cases}
s^{1/2+2\delta},\quad & p  =   N,
\\
s^{2\delta},\quad & p  =   N-4,
\end{cases}
\quad \text{ in } \MME_{[s_0,s_1]}.
\endaligned
$$
We emphasize that, for convenience in the discussion, the left-hand side is stated here in terms of a norm at the order $(p-4)$. 
Thanks to the pointwise decay property stated in Proposition \ref{lem 1 d-KG-e} in combination with the bootstrap bounds
  \eqref{eq h bootstrap}-\eqref{eq l bootstrap}, we find 
\begin{equation}\label{eq2-23-11-2022-M}
|\phi|_{p-4}\lesssim 
\begin{cases}
C_1\eps \la r\ra^{-2}\crochet^{1-\kappa} s^{1/2+2\delta},\quad & p= N,
\\
C_1\eps \la r\ra^{-2}\crochet^{1-\kappa} s^{2\delta},\quad & p= N - 4, 
\end{cases}
\quad \text{ in } \MME_{[s_0,s_1]}
\end{equation}
and, consequently, on the light cone 
\begin{equation}\label{eq3-23-11-2022-M}
|\phi|_p\lesssim
\begin{cases}
C_1\eps \, t^{-7/4 + \delta},\quad & p= N-4,
\\
C_1\eps \, t^{-2 + \delta},\quad & p= N - 8,
\end{cases}
\quad \text{on } \Lscr_{[s_0,s]}.
\end{equation}


\paragraph{Pointwise decay for the wave component.}

We now apply Kirchhoff formula. We have 
$$
|\del_{\alpha}\phi \del_{\beta}\phi|_{N-4} + |\phi^2|_{N-4} \lesssim (C_1\eps)^2\la r\ra^{-2-(1/2-(3/2)\delta)}\crochet^{-1 + 2(1-\kappa)}
\quad \text{ in } \MME_{[s_0,s_1]},
$$
where we used (provided $N\geq 5$)
$$
|\del \phi|_{N-5}\lesssim C_1\eps \la r\ra^{-2}\crochet^{1-\kappa}s^{1/2+2\delta},\quad 
|\del \phi|_{N-3}\lesssim C_1\vep \la r\ra^{-1}\crochet^{-\kappa} s^{1/2+\delta}
\quad \text{ in } \MME_{[s_0,s_1]}.
$$
Therefore, we have  the following low-order, pointwise bound on the source-term of the wave equation: 
\begin{equation}\label{eq1-24-11-2022-M}
|\Box u|_{N-4}\lesssim (C_1\eps)^2\la r\ra^{-2-(1/2-(3/2)\delta)}\crochet^{-1 + 2(1-\kappa)}
\quad \text{ in } \MME_{[s_0,s_1]}.
\end{equation}
In turn, we can apply Proposition~\ref{Linfini wave} (Case 2 therein) within the domain $\MME_{[s_0,s_1]}$, by observing that
 this region is ``past complete'' in the sense that for any $(t,x)\in \MME_{[s_0,s_1]}$ we have $\Lambda_{t,x}\subset \MME_{[s_0,s_1]}$. Therefore, by recalling our assumption \eqref{eq172-bis-22} concerning the sharp pointwise decay of the initial data, we arrive at the following key statement. 

\begin{lemma}
The wave field satisfies the low-order, sharp pointwise estimate
\begin{equation}\label{eq5-19-11-2022-M}
|u|_{N-4}\lesssim (C_1\eps)^2\la r\ra^{-1} + C_0\eps \, t^{-1}
\quad \text{ in } \MME_{[s_0,s_1]}.
\end{equation} 
\end{lemma}

For points $(t,x)\in\MH_{[s_0,s_1]}$, the bound \eqref{eq1-24-11-2022-M} 
and Case 2 of Proposition~\ref{Linfini wave} 
 are also relevant to control the  contribution of the source due to the Euclidean-merging domain, but
  of course we still need to control the contribution due to the hyperboloidal domain, and we are going to apply Case 0 of  Proposition~\ref{Linfini wave}; cf.~\eqref{eq3-24-11-2022-M} below. 


\paragraph{Estimates for source-terms.}

We now derive $L^2$ estimates for the source terms (of the wave equation) and the commutators (of the Klein-Gordon equation). To this end, we want to apply the commutator estimate \eqref{eq7-14-03-2021} with $H = H^{\alpha\beta}u$ and 
the function $u$ therein chosen to be the Klein-Gordon field $\phi$. 
In \eqref{eq7-14-03-2021}, the operator $Z$ is of order $p=\ord(Z) \leq N$ and $k=\rank(Z) \leq p$. 
In view of the Sobolev bound
\eqref{eq2-19-11-2022-M} (low-order case therein) 
and the pointwise Klein-Gordon decay \eqref{eq2-23-11-2022-M}, 
and the sharp pointwise estimate \eqref{eq5-19-11-2022-M}, 
we find 
$
|u||\del\del \phi|_{p-1,k-1}\lesssim C_1\eps \, t^{-1}|\del \phi|_{p,k-1}
$ and 
$$
|\LOmega u|_{k_1-1}|\del\del\phi|_{p_2,k_2}\lesssim 
\begin{cases}
C_1\eps \, t^{-1}|\del\phi|_{p,k-1},\quad & k_1\leq N-4,
\\
C_1\eps \la r\ra^{-1}\crochet^{-\kappa}s^{\delta}|u|_{k_1}\quad & k_1\geq N-3, 
\end{cases}
$$
where the notation in \eqref{eq7-14-03-2021} is used. Here we applied the Sobolev bound \eqref{eq2-24-11-2022-M} to deal with the Klein-Gordon term $|\del\del \phi|_{p_2,k_2}$ and we used the conditions $p_2\leq N-8$ and $N\geq 11$. 
 For the last term in the right-hand side of \eqref{eq7-14-03-2021}, we have 
$$
|\del u|_{p_1-1,k_1}|\del\del \phi|_{p_2,k_2}\lesssim 
\begin{cases}
C_1\eps \, t^{-1}|\del \phi|_{p,k},\quad & p_1\leq N-4,
\\
C_1\eps \la r\ra^{-2}\crochet^{1-\kappa}s^{1/2+2\delta}|\del u|_{p,k},\quad & p_2\leq N-6.
\end{cases}
$$
Collecting the bounds above together and integrating over the Euclidean-merging domain, we obtain
\begin{equation}\label{eq8-29-11-2022-M}
\|J\zeta^{-1}\crochet^{\kappa}[Z,H^{\alpha\beta}\del_{\alpha}\del_{\beta}]\phi\|_{L^2(\MME_s)}
\lesssim C_1\eps s^{-1}\Fenergy_{\kappa,c}^{p,k}(s,\phi) + 
\begin{cases}
(C_1\eps)^2 s^{-5/4},\quad& k\leq N-4,
\\
(C_1\eps)^2 s^{-1/2}, & k\geq N-3.
\end{cases}
\end{equation}
On the other hand, for the wave equation we have
\begin{equation}\label{eq4-27-11-2022-M}
\aligned
|\Box u|_{p,k}
& \lesssim |\del_{\alpha}\phi\del_{\beta}\phi|_{p,k} + |\phi^2|_{p,k}
\\
&
\lesssim C_1\eps \la r\ra^{-2}\crochet^{1-\kappa}s^{1/2+2\delta}(|\del\phi|_{p,k} + |\phi|_{p,k})
\quad \text{ in } \MME_{[s_0,s_1]}
\endaligned
\end{equation}
(provided $N\geq 9$) and, after integration, with $p=\ord(Z) \leq N$ and $k=\rank(Z) \leq p$
\begin{equation}\label{eq1-01-12-2022-M}
\|J\zeta^{-1}\crochet^{\kappa}|\Box u|_{p,k}\|_{L^2(\MME_s)}
\lesssim (C_1\eps)^2s^{-3/2+3\delta}. 
\end{equation}
 

\subsection{Sharp decay in the hyperboloidal domain: hierarchy inequalities}
\label{section----64}

\paragraph{Sharp decay of the wave component.} 

We now rely on Kirchhoff formula, that is, on Proposition~\ref{Linfini wave} and we thus need first to control $|\Box u|$. At this stage, \eqref{eq1-24-11-2022-M} already gives us the required control in $\MME_{[s_0,s_1]}$, and we can focus on the domain $\MH_{[s_0,s_1]}$. The Sobolev decay bound \eqref{eq2-24-11-2022-M} gives us 
\begin{equation}
|\Box u|_{N-6}\lesssim (C_1\eps)^2 t^{-2 -(1/2-\delta)}(t-r)^{-1+ (1/2+\delta)}
\quad \text{ in } \MH_{[s_0,s_1]}, 
\end{equation}
where we used that $|\del\phi\del\phi|_{N-6}\lesssim |\del \phi|_{N-6}|\phi|_{N-6}$
(provided $N\geq 7$).
In view of the estimate in Case 0 of Proposition~\ref{Linfini wave} and thanks to the sharp initial decay condition 
\eqref{eq172-bis-22}, we obtain the following almost sharp pointwise decay of the wave component at a low order: 
\begin{equation}\label{eq3-24-11-2022-M}
|u|_{N-6}\lesssim C_1\eps \, t^{-1+2\delta}
\quad \text{ in } \MH_{[s_0,s_1]}.
\end{equation}


\paragraph{Sharp decay of the Klein-Gordon component.}

On the other hand, we can apply the integration technique in Proposition \ref{prop1-23-11-2022-M} to the Klein-Gordon equation
$
-\Boxt_gZ\phi + c^2\phi = [Z,H^{\alpha\beta}\del_{\alpha}\del_{\beta}]\phi, 
$
in which 
$g^{\alpha\beta} = \eta^{\alpha\beta} + H^{\alpha\beta}$
and $\Boxt_g = g^{\alpha\beta}\del_{\alpha}\del_{\beta}$. 
We pick up any operator $Z$ such that 
$\ord(Z) = p\leq N-4$ and $\rank(Z) = k \leq p$. 
With the notation $|\RR_g[Z\phi]|$ given in \eqref{eq11-23-11-2022-M} and thanks to the Sobolev decay estimates, we claim that 
\begin{equation}
\aligned
|\RR_g[Z\phi]| 
& \lesssim
s^{-2}|\phi|_{N-2} + (t/s)^2s^{-2}|u|_{N-4}|\phi|_{N-2} 
+ (t/s^2)\sum_{p_1+p_2=N-4}|u|_{p_1}|\del\phi|_{p_2+1}
\\
& \lesssim (C_1\eps)(s/t)s^{-3+3\delta}
\qquad \text{ in } \MH_{[s_0,s_1]}.
\endaligned
\end{equation}
Indeed, this is checked as follow. For the first two terms in \eqref{eq11-23-11-2022-M} we apply the Sobolev decay \eqref{eq2-24-11-2022-M} and \eqref{eq4-24-11-2022-M}, and we obtain the upper bound $(C_1\eps)^2(s/t)^{3/2}s^{-3+2\delta}$.  
On the other hand, dealing with the last term  in \eqref{eq11-23-11-2022-M}  is more involved, and we 
observe that  $p_1\leq N-6$ or $p_2\leq N-7$ (with $N\geq 9$) and we write
$$
\aligned
&(t/s^2)|u|_{N-4}|\del\phi|_{N-6}\lesssim (C_1\eps)^2(s/t)s^{-2+2\delta}
\quad 
& \text{ in } \MH_{[s_0,s_1]}, 
\\
&(t/s^2)|u|_{N-6}|\del\phi|_{N-3}\lesssim (C_1\eps)^2(s/t)^{3/2-2\delta}s^{-3+3\delta} \quad 
& \text{ in } \MH_{[s_0,s_1]}, 
\endaligned
$$ 
where for the second case we use the decay \eqref{eq3-24-11-2022-M}.

Next let us consider the boundary ``sup'' term in \eqref{eq5-23-11-2022-M}. In view of the 
light cone estimates \eqref{eq1-23-11-2022-M} and \eqref{eq3-23-11-2022-M} and recalling $\kappa-3/4\geq \delta$, 
we find the following estimate on the boundary of the hyperboloidal domain: 
\begin{equation}\label{eq10-23-11-2022-M}
t^{1/4} \big( t^{1/2}|\phi|_p + |\del\phi|_p + t \, |\delsN\phi|_p \big) 
\lesssim 
\begin{cases}
C_1\eps  (s/t)^{1-\delta},\quad & p=N-4,
\\
C_1\eps (s/t)^{2-2\delta}, \quad & p=N-5.
\end{cases}\quad \text{on } \MH_{s_0}\cup \Lcal_{[s_0,s_1]}.
\end{equation}
We now recall \eqref{eq5-23-11-2022-M} together with \eqref{eq10-23-11-2022-M}, \eqref{eq6-23-11-2022-M} and \eqref{eq9-23-11-2022-M}, since we still need to bound the source term of the Klein-Gordon equation together with 
the additional term $\RR_g[\phi]$. We observe that, for a high-order operator $Z$ with $\ord(Z) = p\leq N-5$ and $\rank(Z) = k$ and thanks to \eqref{eq7-14-03-2021}, 
$$
\aligned
& \big|[Z,H^{\alpha\beta}u\del_{\alpha}\del_{\beta}]\phi\big| 
\\
& \lesssim  |u||\del\del\phi|_{p-1,k-1} 
+ \sum_{k_1+p_2=p\atop k_1+k_2=k}|\LOmega u|_{k_1-1}|\del\del\phi|_{p_2,k_2} 
+ \sum_{p_1+p_2=p\atop k_1+k_2=k}|\del u|_{p_1-1,k_1}|\del\del\phi|_{p_2,k_2}
\\
& \lesssim   (s/t)^{2-2\delta}s^{-5/2} \Abf^{\flat}_{k-1}(s)\Bbf_0(s) 
+ (s/t)^{2-2\delta}s^{-5/2} \sum_{k_1=1}^k \Abf^{\flat}_{k-k_1}(s)\Bbf_{k_1}(s)
\\
& \quad + C_1\eps (s/t)^{3/2-2\delta}s^{-3+\delta}\Abf^{\flat}_{k}(s),
\endaligned
$$
where  
we have introduced
$$
\aligned
& \Abf^{\flat}_k(s) := \sup_{\MH_{[s_0,s_1]}} \Big(
(s/t)^{-2+2\delta}s^{3/2}\big(|\phi|_{N-5,k} + (s/t) \, |\del \phi|_{N-5,k}\big)\Big),
\\
& \Bbf_k(s) := \sup_{\MH_{[s_0,s_1]}}\big( t \, |u|_{N-4,k}\big).
\endaligned
$$
Importantly, in the above estimate for $[Z,H^{\alpha\beta}u\del_{\alpha}\del_{\beta}]\phi$
the first two terms in the right-hand side do not exist when $k=0$. 
Consequently, by collecting the previous estimates together we are in a position to apply Proposition~\ref{prop1-23-11-2022-M} and obtain 
our inductive inequality (at the order $N-5$) concerning the Klein-Gordon component
\begin{equation}\label{eq1-26-11-2022-M}
\aligned
& 
\Abf^{\flat}_{k}(s) \lesssim  C_1\eps + C_1\eps(s/t)^{3/2-2\delta}\int_{\lambda_0}^s\lambda^{-3/2+\delta}\Abf^{\flat}_{k}(\lambda)d\lambda
\\
&  + (s/t)^{2-2\delta}\int_{\lambda_0}^{s} \lambda^{-1}\Abf^{\flat}_{k-1}(\lambda)\Bbf_0(\lambda)d\lambda
 + (s/t)^{2-2\delta}\sum_{k_1=1}^{N-5}\int_{\lambda_0}^s\lambda^{-1} \Abf^{\flat}_{k-k_1}(\lambda)\Bbf_{k_1}(\lambda)d\lambda,
\endaligned
\end{equation}
where the last two terms do not exist when $k=0$.

Following the same lines and setting 
\be
\Abf^{\sharp}_k := \sup_{\MH_{[s_0,s_1]}}\big(
(s/t)^{-1+\delta}s^{3/2}\big(|\phi|_{N-4,k} + (s/t) \, |\del \phi|_{N-4,k}\big)
\big),
\ee
we proceed with similar quantities with $N-5$ replaced by $N-4$ and, by now including a defavorable factor $(s/t)$,  
we find our inductive inequality (at the order $N-4$) 
\begin{equation}\label{eq1-26-11-2022-Mdeux}
\aligned 
&
\Abf^{\sharp}_{k}(s) 
  \lesssim  
C_1\eps + (s/t)^{1/2-\delta}\int_{\lambda_0}^s\lambda^{-3/2+\delta}\Abf^{\sharp}_{k}(\lambda)d\lambda
\\
&   
 + (s/t)^{1-\delta}\int_{\lambda_0}^{s} \lambda^{-1}\Abf^{\sharp}_{k-1}(\lambda)\Bbf_0(\lambda)d\lambda
 + (s/t)^{1-\delta}\sum_{k_1=1}^{N-4}\int_{\lambda_0}^s\lambda^{-1} \Abf^{\sharp}_{k-k_1}(\lambda)\Bbf_{k_1}(\lambda)d\lambda, 
\endaligned
\end{equation}
where the last two terms do not exist when $k=0$.

Finally, considering the case $k=0$ and applying Gronwall inequality to \eqref{eq1-26-11-2022-M}-\eqref{eq1-26-11-2022-Mdeux}, we arrive at the following statement. Larger values of $k$ will be treated next. 

\begin{lemma}[Version without boosts and rotations]
 In the hyperboloidal domain, the Klein-Gordon field satisfies the following 
sharp pointwise estimates at low-order: 
\begin{equation}\label{eq3-26-11-2022-M}
\aligned
&|\phi|_{N-5,0} + (s/t) \, |\del \phi|_{N-5,0} \lesssim C_1\eps (s/t)^{2-2\delta}s^{-3/2}
\quad 
&&\text{ in } \MH_{[s_0,s_1]}, 
\\
&|\phi|_{N-4,0} + (s/t) \, |\del \phi|_{N-4,0} \lesssim C_1\eps (s/t)^{1-\delta}s^{3/2}
\quad 
&&\text{ in } \MH_{[s_0,s_1]}. 
\endaligned
\end{equation}
\end{lemma} 


\paragraph{Sharp decay of the wave component.}

In order to now apply Proposition~\ref{Linfini wave}, we estimate $|\Box u|_k$ for $k\leq N-5$. In fact, 
in $\MH_{[s_0,s]}$ we have 
$$
|\del\phi\del\phi|_{k}\lesssim \sum_{p_1+p_2=k\atop k_1+k_2=k}|\del \phi|_{p_1,k_2} |\del \phi|_{p_2,k_2}
\lesssim (s/t)^{2-4\delta}s^{-3} \sum_{k_1=0}^{N-5}\Abf^{\flat}_{k_1}(s)\Abf^{\flat}_{N-5-k_1}(s).
$$
Similar estimate holds for $|\phi^2|_{N-5,k}$. We observe \eqref{eq1-24-11-2022-M} provides us with
a sufficient estimate on $|\Box u|_k$ in $\MME_{[s_0,s]}$. Thus by Proposition~\ref{Linfini wave} 
(Case 0 and Case 1 therein) together with the pointwise decay on the initial data \eqref{eq172-bis-22}, we find 
\begin{equation}\label{eq2-26-11-2022-M}
\textstyle 
\Bbf_k(s) \lesssim C_0\eps + \sum_{k_1=0}^{N-5}\Abf^{\flat}_{k_1}(s)\Abf^{\flat}_{N-5-k_1}(s).
\end{equation}
When $k=0$, thanks to \eqref{eq3-26-11-2022-M} this leads us to
$
\Bbf_0(s)\lesssim C_1\eps.
$

 
\subsection{Sharp decay in the hyperboloidal domain: induction argument}

We consider \eqref{eq1-26-11-2022-M} together with \eqref{eq2-26-11-2022-M}. By induction on $k$ and Gronwall inequality, provided $C_1\eps \lesssim \theta\leq \delta/(10N)$ we obtain 
\begin{equation}\label{eq4-26-11-2022-M}
\Abf^{\flat}_{k}(s) \lesssim C_1\eps s^{2k\theta},\qquad \Bbf_k(s) \lesssim C_1\eps s^{2k\theta},\qquad 0\leq k\leq N-5.
 \end{equation}
Then we turn to the control of $\Abf^{\sharp}_k$ and $\Bbf_{N-4}$. To this end we need a more detailed analysis of the structure of commutators. Recalling the hierarchy structure \eqref{eq7-14-03-2021}, in the hyperboloidal domain and for 
all $\ord(Z) = p= N-4$ and $\rank(Z) = k\leq N-4$ we find 
\begin{equation}\label{eq5-26-11-2022-M}
\aligned
& \, \big|[Z,H^{\alpha\beta}u\del_{\alpha}\del_{\beta}]\phi\big| 
\\
& \lesssim  |u||\del\del\phi|_{p-1,k-1} 
+ \sum_{k_1+p_2=p\atop k_1+k_2=k}|\LOmega u|_{k_1-1}|\del\del\phi|_{p_2,k_2} 
+ \sum_{p_1+p_2=p\atop k_1+k_2=k}|\del u|_{p_1-1,k_1}|\del\del\phi|_{p_2,k_2}
\\
& \lesssim  (s/t)^{1-\delta}s^{-3/2}\Abf^{\sharp}_{k-1}(s) \Bbf_0(s) 
   + |\LOmega u||\del\del \phi|_{p-1,k-1} 
+ \hskip-.3cm \sum_{k_1+p_2=p,k_1\geq 2\atop k_1+k_2=k} \hskip-.3cm
|\LOmega u|_{k_1-1}|\del\del \phi|_{p_2,k_2} 
\\
& \quad + C_1\eps (s/t)^{1/2-\delta}s^{-3+\delta}\Abf^{\sharp}_k(s)
\\
& \lesssim  (s/t)^{1-\delta}s^{-3/2}\Abf^{\sharp}_{k-1}(s) \big(\Bbf_0(s) + \Bbf_1(s)\big) 
+ (s/t)^{2-2\delta}s^{-5/2}\sum_{k_1=2}^{k} \Abf^{\flat}_{k-k_1}(s)\Bbf_{k_1}(s)
\\
& \quad + C_1\eps (s/t)^{1/2-\delta}s^{-3+\delta}\Abf^{\sharp}_k(s).
\endaligned
\end{equation}
Then, thanks to \eqref{eq4-26-11-2022-M}, for all $k\leq N-5$ we have 
\begin{equation}
\aligned
\big|[Z,H^{\alpha\beta}\del_{\alpha}\del_{\beta}]\phi\big|
& \lesssim  C_1\eps (s/t)^{1/2-\delta}s^{-3+\delta}\Abf^{\sharp}_k(s)  
+ C_1\eps(s/t)^{1-\delta}s^{-5/2+2\theta}\Abf^{\sharp}_{k-1}(s) 
\\
& \quad + (C_1\eps)^2(s/t)^{2-2\delta}s^{-5/2+2k\theta} .
\endaligned
\end{equation}
When $k=N-4$, the second sum in the right-hand side of \eqref{eq5-26-11-2022-M} 
can be decomposed and controlled as follows: 
$$
\aligned
& \textstyle
(s/t)^{2-2\delta}s^{-5/2}\sum_{k_1=2}^{N-4} \Abf^{\flat}_{N-4-k_1}(s)\Bbf_{k_1}(s) 
\\
& \lesssim  (s/t)^{2-2\delta}s^{-5/2}\Abf^{\flat}_0(s)\Bbf_{N-4}(s) + (C_1\eps)^2(s/t)^{2-2\delta}s^{-5/2}s^{2(N-4)\theta}
\\
& \lesssim  C_1\eps (s/t)^{2-2\delta}s^{-5/2} \Bbf_{N-4}(s) + (C_1\eps)^2(s/t)^{2-2\delta}s^{-5/2}s^{2(N-4)\theta}.
\endaligned
$$

We are now ready to apply  Proposition~\ref{prop1-23-11-2022-M} concerning the sharp decay of Klein-Gordon solutions. 
Thanks to  \eqref{eq10-23-11-2022-M}, \eqref{eq6-23-11-2022-M}, and \eqref{eq9-23-11-2022-M},
 for all $k\leq N-5$ we find 
\begin{equation}
\aligned
\Abf^{\sharp}_k(s)& \lesssim  C_1\eps s^{2k\theta} 
+ C_1\eps(s/t)^{1/2-\delta} \int_{\lambda_0}^s\lambda^{-3/2}\Abf^{\sharp}_k(\lambda)d\lambda 
\\
& \quad + C_1\eps(s/t)^{1-\delta} \int_{\lambda_0}^s \lambda^{-1+2\theta}\Abf^{\sharp}_{k-1}(\lambda)d\lambda.
\endaligned
\end{equation}
By induction on $k$ and Gronwall's inequality, provided $C_1\eps \theta\lesssim 1$ we obtain
\begin{equation}
\Abf^{\sharp}_k(s)\lesssim C_1\eps s^{2k\theta},\quad k\leq N-5.
\end{equation}
On the other hand, we have 
\begin{equation}\label{eq7-26-11-2022-M}
\aligned
\Abf^{\sharp}_{N-4}(s)& \lesssim  C_1\eps s^{2(N-4)\theta} 
+ C_1\eps(s/t)^{1/2-\delta} \int_{\lambda_0}^s\lambda^{-3/2}\Abf^{\sharp}_{N-4}(\lambda)d\lambda 
\\
& \quad +C_1\eps (s/t)^{2-2\delta}\int_{\lambda_0}^s\lambda^{-1}\Bbf_{N-4}(\lambda)d\lambda.
\endaligned
\end{equation}
A final estimate on $\Bbf_{N-4}$ is required and we write 
$$
\aligned
|\del \phi\del \phi|_{N-4} &
\textstyle \lesssim  |\del\phi||\del\phi|_{N-4} + \sum_{k=1}^{N-5}|\del\phi|_k|\del\phi|_{N-4-k}
\\
& \textstyle
\lesssim  (s/t)^{-\delta}s^{-3/2}|\phi|_{1,0} \Abf^{\sharp}_{N-4}(s) + \sum_{k=1}^{N-5}\Abf^{\flat}_k(s)\Abf^{\flat}_{N-4-k}(s)
\\
& \lesssim  (s/t)^{-\delta}s^{-3/2}|\phi|_{N-5,0}\Abf^{\sharp}_{N-4}(s) + (C_1\eps)^2(s/t)^{2-4\delta}s^{-3 + 2(N-4)\theta}
\\
& \lesssim   C_1\eps (s/t)^{2-3\delta}s^{-3}\Abf^{\sharp}_{N-4}(s) + (C_1\eps)^2(s/t)^{2-4\delta}s^{-3 + 2(N-4)\theta}
\endaligned
$$
(provided $N\geq 6$). A similar (but simpler) estimate holds for $|\phi|^2$. Then from Proposition~\ref{Linfini wave} concerning the wave equation
(Case 0 and Case 1 therein), we deduce that 
\begin{equation}\label{eq6-26-11-2022-M}
\Bbf_{N-4}(s)\lesssim  C_0\eps  + C_1\eps \Abf^{\sharp}_{N-4}(s) + (C_1\eps)^2s^{2(N-4)\theta}.
\end{equation}
Combining this inequality \eqref{eq6-26-11-2022-M} together with \eqref{eq7-26-11-2022-M}, we  
 apply Gronwall's inequality and obtain
\begin{equation}\label{eq8-26-11-2022-M}
\Abf^{\sharp}_{N-4}(s)\lesssim C_1\eps s^{(N-4)\theta},
\qquad 
\Bbf_{N-4}(s) = C_1\eps s^{2(N-4)\theta}.
\end{equation}
We have arrived at the following main conclusion.

\begin{lemma}[Version with a hierarchy]
\label{prop1-27-11-2022-M}
 In the hyperboloidal domain $\MH_{[s_0,s_1]}$, the wave and Klein-Gordon fields satisfy the
following sharp pointwise estimates at low-order
(provided $C_1\eps \lesssim \theta \leq \delta/(10N)$):  
\begin{equation}\label{eq5-27-11-2022-M}
t \, |u|_{k}\lesssim C_1\eps s^{2k\theta},\qquad k\leq N-4, 
\end{equation}
\bse
\begin{equation}\label{eq6-27-11-2022-M}
|\phi|_{N-5,k} + (s/t) \, |\del\phi|_{N-5,k}\lesssim C_1\eps (s/t)^{2-2\delta}s^{-3/2+2k\theta},\qquad k\leq N-5,
\end{equation}
\begin{equation}
|\phi|_{N-4,k} + (s/t) \, |\del\phi|_{N-4,k}\lesssim C_1\eps (s/t)^{1-\delta}s^{-3/2+2k\theta},\qquad k\leq N-4.
\end{equation}
\ese
\end{lemma}


\subsection{Commutators and source terms in the hyperboloidal domain}

We now turn our attention to the $L^2$-type estimates of the commutators and source terms.
First of all, to deal with the Klein-Gordon commutators, we apply the commutator estimate \eqref{eq7-14-03-2021}. For the first term in the right-hand side therein, we have
$$
|u||\del\del \phi|_{p-1,k-1}\lesssim C_1\eps s^{-1}(s/t) \, |\del\phi|_{p,k-1}.
$$
For the second term,  when $k_1\leq N-4$ we apply \eqref{eq5-27-11-2022-M} and obtain  
$$
|\LOmega u|_{k_1-1}|\del\del\phi|_{p_2,k_2}\lesssim C_1\eps  s^{-1+2k_1\theta}(s/t) \, |\del\phi|_{p,k-k_1}.
$$
The case $N-3\leq k_1\leq k$ occurs {\sl only when} $p\geq k\geq N-3$. In this case $p_2 = p-k_1\leq 3$ and 
$$
\aligned
|\LOmega u|_{k_1-1}|\del\del\phi|_{p_2,k_2}
& \lesssim  C_1\eps (s/t)^{2-2\delta} s^{-3/2+2k_2\theta}|\LOmega u|_{k_1-1} 
\\
& \lesssim  C_1\eps (s/t)^{1-2\delta} s^{-1/2+2k_2\theta}|\delus u|_{k_1-1},
\endaligned
$$
where we applied \eqref{eq6-27-11-2022-M} (provided $3\leq N-7$, namely $N\geq 10$).

For the last term in the right-hand side of \eqref{eq7-14-03-2021}, we observe
 that when $p_2\leq N-7$ we can apply \eqref{eq6-27-11-2022-M} on $|\del\del\phi|_{p_2,k_2}\leq |\phi|_{N-5,k_2}\lesssim C_1\eps (s/t)^{2-\delta}s^{-3/2+(N-5)\theta}$:
$$
|\del u|_{p_1-1}|\del\del \phi|_{p_2,k_2}\lesssim C_1\eps s^{-5/4}(s/t) \, |\del u|_{p_1-1}.
$$
When $N-2\geq p_2\geq N-6$, in this case $p_1\leq 6\leq N-5$. We apply \eqref{eq7-23-11-2022-M} to the term $|\del u|_{p_1-1}$ and we find
$$
|\del u|_{p_1-1}|\del\del \phi|_{p_2,k_2}\lesssim C_1\eps s^{-1} (s/t) \, |\del \phi|_{p,k}.
$$
In the hyperboloidal domain, for all $\ord(Z) = p$ and $\rank(Z) = k$ we conclude that
$$
\aligned
 \, \big|[Z,H^{\alpha\beta}\del_{\alpha}\del_{\beta}]\phi\big|
& \lesssim  C_1\eps s^{-1} (s/t) \, |\del u|_{p,k} + C_1\eps (s/t)^{-3/2+2(N-5)\theta}|\del u|_{p,k}
\\
& \quad +
\begin{cases}
C_1\eps \sum_{k_1=1}^k s^{-1+2k_1\theta}(s/t) \, |\del\phi|_{p,k-k_1},\quad &k\leq N-4,
\\
{
\aligned
& C_1\eps \sum_{k_1=1}^{N-4} s^{-1+2k_1\theta}(s/t) \, |\del\phi|_{p,k-k_1}
\\
&  +\sum_{k_1=0}^{k-N+3}s^{-1/2+2k_1\theta}|\delus u|_{k-k_1-1},
\endaligned 
}& k\geq N-3.
\end{cases}
\endaligned
$$
Provided that $2(N-5)\theta + \delta\leq 1/4$, for the $L^2$ norm this pointwise inequality  leads us to
\begin{equation}\label{eq3-30-11-2022-M}
\aligned
&
\|[Z,H^{\alpha\beta}\del_{\alpha}\del_{\beta}]\phi\|_{L^2(\MH_s)}
 \lesssim  (C_1\eps)^2s^{-5/4}
\\
& \quad +
\begin{cases}
C_1\eps s^{-1+2k_1\theta} \sum_{k_1=0}^k\Fenergy_{\kappa,c}^{p,k-k_1}(s,\phi),\quad & k\leq N-4,
\\
\\
{
\aligned
& \textstyle 
C_1\eps s^{-1+2k_1\theta} \sum_{k_1=0}^{N-4}\Fenergy_{\kappa,c}^{p,k-k_1}(s,\phi)
\\
\\
& \textstyle \quad + \sum_{k_1=0}^{k-N+3} s^{-1/2+2k_1\theta}\Fenergy_{\kappa,c}^{k-k_1-1}(s,u),
\endaligned
}&k\geq N-3
\end{cases}
\endaligned
\end{equation}
On the other hand, dealing with the source term for the wave equation is simpler and we write directly
\begin{equation}\label{eq1-30-11-2022-M}
\||\Box u|_{p,k}\|_{L^2(\MH_s)}
\lesssim 
\begin{cases}
(C_1\eps)^2 s^{-3/2+2\delta},\quad & p\leq N-4,
\\
(C_1\eps)\sum_{k_1=0}^ks^{-3/2+2k_1\theta}\Fenergy_{\kappa,c}^{p,k-k_1}(s,\phi),
& p\geq N-3.
\end{cases}
\end{equation}


\subsection{Improved energy estimate and conclusion}
\label{section---67}

\paragraph{Preliminaries.}

Now we are ready to apply the energy estimate \eqref{eq1-27-11-2022-M} and establish the following {\sl improved energy estimates}:
\begin{equation}\label{eq2-27-11-2022-M}
s^{-\delta} \, \Fenergy^N_\kappa(s,u)
+ s^{-1/2-\delta} \, \Fenergy^N_{c, \kappa}(s, \phi)
\leq (C_1/2) \eps,
\end{equation}
\begin{equation}\label{eq3-27-11-2022-M}
\Fenergy_\kappa^{N-4}(s,u) + s^{-\delta} \Fenergy^{N-4}_{c, \kappa}(s, \phi)
\leq (C_1/2) \eps.
\end{equation} 
We need to distinguish between high-order case $p\geq N-3$ and low-order case $p\leq N-4$. But before these, we establish the common estimates needed for both.

\begin{lemma}\label{prop1-29-11-2022-M}
In the spacetime slab $\Mscr_{[s_0,s_1]}$, the following estimates hold for $\ord(Z) = p$ and $\rank(Z) = k$:
\begin{equation}\label{eq2-29-11-2022-M}
\int_{\Mscr_s}|G_{g,\kappa}[Z\phi]|\ Jdxds\lesssim 
C_1\eps s^{-1}\Eenergy_{\kappa}^{p,k}(s,\phi)ds,
\end{equation} 
\begin{equation}\label{eq3-29-11-2022-M}
\int_{\Mscr_s}\crochet^{2\kappa-1}\aleph'(r-t)H^{\N00}|\del_tZ\phi|\ Jdxds \lesssim  
C_1\eps s^{-1}\Eenergy_{\kappa}^{p,k}(s,\phi)ds.
\end{equation}
Here we have $g^{\alpha\beta} = \eta^{\alpha\beta} + H^{\alpha\beta}u$, while the notation 
$G_{g,\kappa}$ was introduced in \eqref{eq8-27-11-2022-M}. Furthermore, when $C_1\eps$ sufficiently small one also has 
\begin{equation}\label{eq4-29-11-2022-M}
\frac{1}{4} \, \Eenergy_{g,\kappa,c}(s,\phi)\leq \Eenergy_{\kappa,c}(s,\phi)\leq 4 \, \Eenergy_{g,\kappa,c}(s,\phi). 
\end{equation}
\end{lemma}

\begin{proof}
We rely here on direct consequences of the pointwise decay derived earlier. Thanks to \eqref{eq2-19-11-2022-M} and \eqref{eq7-23-11-2022-M}, we have 
$$
\aligned
|G_{g,\kappa}[Z\phi]| J \lesssim |\del u| \crochet^{2\kappa}|\del Z\phi| 
& \lesssim 
\begin{cases}
C_1\eps \, t^{-1/2}s^{-1}(s/t)^{-1} |(s/t)\del Z\phi|^2 
\quad & \text{ in }\MH_{[s_0,s_1]},
\\
C_1\eps \, t^{-1}\crochet^{-\kappa} (J\zeta^{-2}) |\zeta\crochet^{\kappa}\del Z \phi|^2
\quad & \text{ in }\MME_{[s_0,s_1]},
\end{cases}
\\
& \lesssim 
\begin{cases}
C_1\eps s^{-1} |(s/t)\del Z\phi|^2
\quad & \text{ in }\MH_{[s_0,s_1]},
\\
C_1\eps s^{-1} |\zeta\crochet^{\kappa}\del Z \phi|^2
\quad & \text{ in }\MME_{[s_0,s_1]},
\end{cases}
\endaligned
$$
which, by integration, implies \eqref{eq2-29-11-2022-M}. On the other hand, the proof of \eqref{eq3-29-11-2022-M} is similar and 
we apply \eqref{eq5-19-11-2022-M} (in the Euclidean-merging domain) and \eqref{eq5-27-11-2022-M} (in the hyperboloidal domain).
Concerning the inequalities \eqref{eq4-29-11-2022-M}, we observe that
$
|\Eenergy_{g,\kappa,c}(s,Z\phi) - \Eenergy_{\kappa,c}(s,Z\phi)|
\lesssim \int_{\Mscr_s} \zeta^{-2}|u| \, |\zeta\crochet^{\kappa} \del Z\phi|^2 \, dx
$
and, in view of \eqref{eq5-19-11-2022-M} and \eqref{eq5-27-11-2022-M}, we find (since $t^{-1}\lesssim \zeta^2$) 
$$
|\Eenergy_{g,\kappa,c}(s,Z\phi) - \Eenergy_{\kappa,c}(s,Z\phi)|
\lesssim C_1\eps\int_{\Mscr_s} \Eenergy_{\kappa,c}(s,Z\phi).
\qedhere 
$$ 
\end{proof}


\paragraph{Energy inequalities.}

Now we apply the energy equation \eqref{eq1-27-11-2022-M}. For the wave operator, the metric under consideration in the model is flat, thus for all $\ord(Z) = p$ and $\rank(Z) = k$ 
$$ 
\frac{d}{ds}\Eenergy_{\kappa}(s,Zu) 
= \int_{\Mscr_s} \hskip-.3cm
\zeta\crochet^{2\kappa}\del_t Zu\,\crochet^{\kappa}J\zeta^{-1} \Box Zu\ dx
  \lesssim  \Fenergy_{\kappa}(s,Zu)^{1/2}\|J\zeta^{-1}\,\crochet^{\kappa} |\Box u|_{p,k}\|_{L^2(\Mscr_s)}.
$$
Thus we have $
\frac{d}{ds}\Fenergy_{\kappa}(s,Zu)\lesssim \|(s/t) \, |\Box u|_{p,k}\|_{L^2(\MH_s)} 
+ \|J\zeta^{-1}\, \crochet^{\kappa}|\Box u|_{p,k}\|_{L^2(\MME_s)}.
$
Thanks to the $L^2$ bound \eqref{eq1-01-12-2022-M} and \eqref{eq1-30-11-2022-M} concerning the source-terms, 
and after summation over $\ord(Z)\leq p$ and $\rank(Z)\leq k$, we find 
\begin{equation}\label{eq2-01-12-2022-M}
\aligned
& \frac{d}{ds}\Fenergy_{\kappa}^{p,k}(s,u)
\\
& \leq
\begin{cases}
C_{N,\delta}(C_1\eps)^2s^{-5/4},
& p\leq N-4,
\\
C_{N,\delta}(C_1\eps)^2s^{-5/4} 
+ C_{N,\delta} C_1\eps \sum_{k_1=0}^k s^{-3/2+2k_1\theta}\Fenergy_{\kappa,c}^{p,k-k_1}(s,\phi),
 & p\geq N-3,
\end{cases}
\endaligned
\end{equation}
where we the notation $C_{N,\delta}$ we now make explicit the implicit constant corresponding by the symbol ``$\lesssim$'' in all previous estimates. This constant is determined by $N,\delta$ and the model itself.

Now we turn our attention to the Klein-Gordon component. Recall \eqref{eq1-27-11-2022-M} together with 
Lemma~\ref{prop1-29-11-2022-M}, for all $\ord(Z)\leq p$ and $\rank(Z)\leq k$ we have 
$$
\frac{d}{ds}\Eenergy_{g,\kappa,c}(s,Z\phi)\lesssim C_1\eps s^{-1} \Eenergy_{\kappa}^{p,k}(s,\phi) + \Eenergy_{g,\kappa,c}(s,Z\phi)^{1/2} \big\|[Z,H^{\alpha\beta}\del_{\alpha}\del_{\beta}]\phi\big\|_{L^2(\Mscr_s)}.
$$
Summing over all $\ord(Z)\leq p$ and $\rank(Z)\leq k$, we find
$$
\frac{d}{ds}\Eenergy_{g,\kappa,c}^{p,k}(s,\phi)\lesssim C_1\eps s^{-1} \Eenergy_{\kappa}^{p,k}(s,\phi) + \Eenergy_{g,\kappa,c}(s,Z\phi)^{1/2} \big\|[Z,H^{\alpha\beta}\del_{\alpha}\del_{\beta}]\phi\big\|_{L^2(\Mscr_s)},
$$
which, thanks to \eqref{eq4-29-11-2022-M}, leads us to
\begin{equation}\label{eq7-29-11-2022-M}
\frac{d}{ds}\Fenergy_{\kappa,c}^{p,k}(s,\phi) \lesssim C_1\eps s^{-1}\Fenergy_{\kappa,c}^{p,k}(s,\phi) + \sum_{\ord(Z)\leq p\atop \rank(Z)\leq k}\|[Z,H^{\alpha\beta}\del_{\alpha}\del_{\beta}]\phi\|_{L^2(\Mscr_s)}.
\end{equation}
Taking into account the commutator estimates in \eqref{eq8-29-11-2022-M} and \eqref{eq3-30-11-2022-M}, we 
thus have 
\begin{equation}\label{eq3-01-12-2022-M}
\frac{d}{ds}\Fenergy_{\kappa,c}^{p,k}(s,\phi)\leq
\begin{cases}
\displaystyle
C_{N,\delta}(C_1\eps)^2 s^{-5/4} 
+C_{N,\delta} C_1\eps \sum_{k_1=0}^ks^{-1+2k_1\theta}\Fenergy_{\kappa,c}^{p,k-k_1}(s,\phi), & k\leq N-4,
\\
{\aligned
&C_{N,\delta}(C_1\eps)^2 s^{-1/2} 
+  C_{N,\delta} C_1\eps\sum_{k_1=0}^{N-4}s^{-1+2k_1\theta}\Fenergy_{\kappa,c}^{p,k-k_1}(s,\phi)
\\
& \quad +C_{N,\delta} C_{N,\delta}\sum_{k_1=0}^{k-N+3} s^{-1/2+2k_1\theta}\Fenergy_{\kappa,c}^{k-k_1-1}(s,u),
\endaligned
}\quad & \hskip-.cm k\geq N-3.
\end{cases}
\end{equation}


\paragraph{Low-order estimates.} 

Consider first the range $p\leq N-4$, therefore $k \leq p \leq N-4$. Thus 
in view of the energy inequality \eqref{eq2-01-12-2022-M} and \eqref{eq3-01-12-2022-M}, we find 
\begin{equation}\label{eq1-02-12-2022-M}
\aligned
\Fenergy_{\kappa}^{N-4}(s,u)
& \leq C_0\eps + C_{N,\delta}(C_1\eps)^2,  &&& k \leq p \leq N-4, 
\\
\Fenergy_{\kappa,c}^{N-4,k}(s\phi)
& \leq (C_0\eps +5 C_{N,\delta}(C_1\eps)^2)s^{(2k+1)\theta}, &&&k \leq p \leq N-4, 
\endaligned
\end{equation}
provided 
\begin{equation}\label{eq4-01-12-2022-M}
C_NC_1\eps \leq \theta,\quad C_0/C_1 + 5C_N C_1\eps \leq 2\theta.
\end{equation}


\paragraph{High-order estimates.}

We consider next the range $p\geq N-3$, and also rely on 
 \eqref{eq2-01-12-2022-M} and \eqref{eq3-01-12-2022-M}. However, the system of inequalities is 
 more involved. We first treat the case $k\leq N-4$, as follows: the system \eqref{eq2-01-12-2022-M} and \eqref{eq3-01-12-2022-M} reduces to
$$
\aligned
\frac{d}{ds}\Fenergy_{\kappa}^{p,k}(s,u)& \leq  C_{N,\delta}(C_1\eps)^2s^{-5/4}
 +  C_{N,\delta}C_1\eps s^{-3/2} \mathcalboondox {A}^{p,k}(s,\phi),
\\
\frac{d}{ds}\Fenergy_{\kappa,c}^{p,k}(s,\phi)\leq & C_{N,\delta} (C_1\eps)^2 s^{-5/4} 
 + C_{N,\delta} C_1\eps s^{-1} \mathcalboondox {A}^{p,k}(s,\phi),
\endaligned
$$
in which 
\be
 \mathcalboondox {A}^{p,k}(s,\phi) :=\sum_{k_1=0}^ks^{2k_1\theta}\Fenergy_{\kappa,c}^{p,k-k_1}(s,\phi). 
\ee
By induction and Gronwall's inequality, we obtain 
\begin{equation}\label{eq2-02-12-2022-M}
\aligned
\Fenergy_{\kappa}^{N,k}(s,u)
& \leq C_0\eps + 5C_{N,\delta}(C_1\eps)^2,
&&& p\geq N-3, 
\\ 
\Fenergy_{\kappa,c}^{N,k}(s,\phi)
& \leq \big(C_0\eps + 5C_{N,\delta}(C_1\eps)^2\big)s^{(2k+1)\theta},
&&& p\geq N-3, 
\endaligned
\end{equation}
provided 
\begin{equation}\label{eq5-01-12-2022-M}
(C_0/C_1) + 5C_{N,\delta}C_1\eps\leq \min\{1/4N, 2\theta\},\quad C_{N,\delta}C_1\eps \leq \theta.
\end{equation}

On the other hand, for the range $k\geq N-3$ we write  
$$
\aligned
\mathcalboondox {A}^{p,k}(s,\phi) =& \sum_{k_1=0}^{N-4}s^{2k_1\theta}\Fenergy_{\kappa,c}^{N,k}(s,\phi) 
+ \sum_{k_1=N-3}^{k}s^{2k_1\theta}\Fenergy_{\kappa,c}^{N,k-k_1}(s,\phi)
\\
& \lesssim  (N-4)(C_0\eps + C_N(C_1\eps)^2)s^{(2N+1)\theta} 
+ \mathcalboondox {B}^k(s,\phi), 
\endaligned
$$
in which 
\be
\mathcalboondox {B}^{k}(s,\phi) := \sum_{k_1=N-3}^{k}s^{2k_1\theta}\Fenergy_{\kappa,c}^{N,k-k_1}(s,\phi).
\ee
Substituting this result into \eqref{eq2-01-12-2022-M} and \eqref{eq3-01-12-2022-M}, 
provided $(2N+1)\theta\leq 1/4$ we obtain 
$$
\aligned
\frac{d}{ds}\Fenergy_{\kappa}^{N,k}(s,u)& \leq  C_{N,\delta}' (C_1\eps)^2s^{-5/4}
+ C_{N,\delta}'C_1\eps s^{-3/2} \mathcalboondox {B}^{k}(s,\phi),
\\
\frac{d}{ds}\Fenergy_{\kappa,c}^{N,k}(s,\phi)& 
\leq  C_{N,\delta}'(C_1\eps)^2 s^{-1/2} 
 + C_{N,\delta}'C_1\eps s^{-1/2}\mathcalboondox {C}^k(s,u) .
\endaligned
$$
where 
\be
\mathcalboondox {C}^k(s,u):= \sum_{k_1=N-3}^{k} s^{2k_1\theta}\Fenergy_{\kappa}^{k-k_1-1}(s,u). 
\ee
We especially observe that, by definition, $\mathcalboondox {C}^{N-3}(s,u) = 0$. Proceeding again by induction, we arrive at
$$
\aligned
\Fenergy_{\kappa}^{N,k}(s,u)
& \leq (C_0\eps + 5C_{N,\delta}'(C_1\eps)^2)s^{2k\theta},\quad
\\
\Fenergy_{\kappa,c}^{N,k}(s,\phi)
& \leq (C_0\eps + 5 C_{N,\delta}'(C_1\eps)^2)s^{1/2+2k\theta},
\endaligned
$$ 
so that
\begin{equation}\label{eq3-02-12-2022-M}
\aligned
\Fenergy_{\kappa}^{N,k}(s,u)& \leq  (C_0\eps + 5C_{N,\delta}'(C_1\eps)^2)s^{2(k-(N-3)+1)\theta},
&&& k\geq N-3, 
\\
\Fenergy_{\kappa}^{N,k}(s,\phi)& \leq  (C_0\eps + 5C_{N,\delta}'(C_1\eps)^2)s^{2(k-(N-3))\theta},
&&& k\geq N-3, 
\endaligned
\end{equation}
provided 
\begin{equation}\label{eq4-02-12-2022-M}
C_0/C_1 + 5C_{N,\delta}'C_1\eps \leq \theta/2,  
\qquad
C_{N,\delta}'C_1\eps\leq \theta.
\end{equation}
If we now impose that the constants satisfy 
\be
2N\theta\leq \delta, \quad C_0/C_1\leq 1/4\theta, \quad \eps\leq \frac{2C_1-C_0}{10\max\big( C_{N,\delta},C_{N,\delta'}\big)},
\ee
together with the previous inequalities 
\eqref{eq4-01-12-2022-M}, \eqref{eq5-01-12-2022-M} and \eqref{eq4-02-12-2022-M}, then \eqref{eq1-02-12-2022-M}, \eqref{eq2-02-12-2022-M} and \eqref{eq3-02-12-2022-M} lead to the conclusion
\eqref{eq2-27-11-2022-M}-\eqref{eq3-27-11-2022-M}. This completes the bootstrap argument.
 

\small

\paragraph{Acknowledgments.}
 
This work was done in parts when PLF was a visiting research fellow at the Courant Institute for Mathematical Sciences, New York University, and a visiting professor at the School of Mathematical Sciences, Fudan University, Shanghai. The work of YM was supported by a Special Financial Grant from the China Postdoctoral Science Foundation under the grant number NSFC 11601414.
 

\vskip.15cm

\bibliography{references}


\pagestyle{plain}

\end{document}

%% file: our-macros-CQG.tex
\newcommand \crochet {\mathbf X}

\newcommand \Fenergy {\mathcalboondox F}           
\newcommand \Eenergy {\mathcalboondox E}    
\newcommand \Time {\mathbf T}

\newcommand \E {\mathcal{E}} 

\newcommand \ME {\mathcal{EM}} 
\newcommand \EM {\mathcal{EM}} 
\renewcommand \H {\mathcal H} 

\newcommand \M {\mathcal M} 
\newcommand \Mcal {\mathcal M} 

\newcommand \N {\mathcal N} 

\newcommand \Lcal {\mathcal L} 

\newcommand \Mscr       {\mathscr M} 
\newcommand \Mext       {\Mscr^{\mathcal E}} 

\newcommand \MH 		{\Mscr^{\mathcal{H}}} 

\newcommand \MME 		{\Mscr^{\mathcal{EM}}}  

\newcommand \MM 		{\Mscr^{\mathcal{M}}} 
\newcommand \Mtran      {\Mscr^{\mathcal{M}}}  

\newcommand \MMEnear  	{\Mscr^{\textbf{near}}} 
 
\newcommand \Mnear {\Mscr^\textbf{near}}
\newcommand \Mfar{\Mscr^\textbf{far}}

\newcommand \near {\textbf{near}}

\newcommand \delH {\del^{\mathcal H}} 
\newcommand \delsH {\slashed \del^{\mathcal H}}
\newcommand \PsiH {{\Psi^{\mathcal H}}{}}

\newcommand {\usH} {\slashed{u}^{\mathcal{H}}}

\newcommand \delN {\del^{\mathcal{N}}}
\newcommand \delsN {\slashed \del^{\mathcal{N}}}
\newcommand \PsiN {{\Psi^{\mathcal N}{}}}  
\newcommand \PhiN {{\Phi^{\mathcal N}}{}}

\newcommand \HN {H^{\mathcal{N}}{}}

\newcommand \delsE {\slashed \del^\E}

\newcommand \delsEH {\slashed {\del}^{\mathcal{EH}}}
\newcommand \delsME {\slashed {\del}^{\mathcal{EM}}} 

\newcommand \rhoH {{\mathbf r}^{\mathcal{H}}}    
\newcommand \rhoE {{\mathbf r}^{\mathcal{E}}}                

\newcommand \Fbb {\mathbb F}

\newcommand \Pbb {\mathbb P}

\newcommand \Qbb {\mathbb Q}

\newcommand \Sbb {\mathbb S} 

\newcommand \SbbME{\mathbb S^{\mathcal{EM}}}

\newcommand \Ebf {\mathbf E}

\newcommand \la \langle
\newcommand \ra \rangle

\newcommand \init {\textbf{init}}

\newcommand \dive {\text{div} \hskip.05cm}

\newcommand \del \partial 

 \newcommand \delu \delH 
 
\newcommand \delEH {\del^{\mathcal{E \hskip-.06cm H}}} 

\newcommand \Hb{\overline H}
 
\newcommand \uts   {{\slashed u}^\Ncal}

\newcommand \delus  \delsH

\newcommand \delts  \delsN

\def\blfootnote{\gdef\@thefnmark{} \@footnotetext}
\newcommand \vep \epsilon

\newcommand \Sch {S}

\newcommand \source{\textbf{sour}}

\newcommand \Boxt {\widetilde \Box}

\newcommand \Lscr{\mathscr{L}}

\newcommand \BoxChapeau {\widehat \Box} 
\newcommand \Hcal {\mathcal H}

\newcommand \Hu {\underline{H}}

\newcommand {\gu}{\underline{g}}

\newcommand \RR{\mathbb{R}}
\newcommand \Tbf	{\mathbf{T}}
\newcommand \Abf	{\mathbf{A}}
\newcommand \Bbf	{\mathbf{B}}

\newcommand {\eps} \epsilon

\let\oldmarginpar\marginpar
\renewcommand\marginpar[1]{\- \oldmarginpar[\raggedleft\footnotesize #1]%
{\raggedright\footnotesize #1}}

\newcommand \TJ T

\newcommand \gd g  
 
\newcommand \Rd  R 
\newcommand \Rbd R 
 
\newcommand \omegad  \omega 
\newcommand \nablabd \nabla  
\newcommand \Gammad \Gamma 
\newcommand \Gammab \Gamma 
\newcommand \Deltab \Delta 

\newcommand \coef \kappa

\newcommand \Ocal {\mathcal O}
\newcommand \Ncal {\mathcal N}

\renewcommand \gu {g^\Hcal{}}

\renewcommand \BoxChapeau {\widetilde \Box}

\renewcommand \Hu {H^\Hcal{}}

\newcommand \Gbf {\mathbf G}

\newcommand \Rbf {\mathbf R}
\newcommand \lbf {\mathbf l}

\newcommand \ord {\textbf{ord}}

\newcommand \Mink {\textbf{Mink}}
\newcommand \rank {\textbf{rank}}
\newcommand \gMink {g_\textbf{Mink}}
\renewcommand \Sch {\textbf{Sch}}

\newcommand \glue {{}}                 

\newcommand \easy {\textbf{easy}}

\newcommand \LOmega {Y_\text{rot}} 

\newcommand \exposant \rho
\newcommand \vecnnu \nu
\newcommand \notrelapse L
